\documentclass[11pt]{article}

\usepackage{graphicx} 

\usepackage{amsfonts,amstext,amsmath,amssymb,amsthm} 
\usepackage[utf8]{inputenc}
\usepackage[T1]{fontenc}

\usepackage{acronym}

\usepackage{titlesec}
\usepackage{etoolbox}
\makeatletter
\patchcmd{\ttlh@hang}{\parindent\z@}{\parindent\z@\leavevmode}{}{}
\patchcmd{\ttlh@hang}{\noindent}{}{}{}
\makeatother

\usepackage{booktabs}

\usepackage[authoryear]{natbib}

\usepackage[]{subfigure}
\usepackage{caption2}
\usepackage{url}
\usepackage{placeins} 

\titleformat{\section}{\large\bfseries}{\thesection.}{.5em}{}
\titlespacing*{\section}{0pt}{*3}{*2}
\titleformat{\subsection}{\normalfont\bfseries}{\thesubsection.}{.5em}{}
\titlespacing*{\subsection} {0pt}{*3}{*2}
\titleformat{\subsubsection}{\normalfont\bfseries}{\thesubsubsection.}{.5em}{}
\titlespacing*{\subsubsection} {0pt}{*3}{*2}

\usepackage{lmodern}

\long\def\symbolfootnote[#1]#2{\begingroup
\def\thefootnote{\fnsymbol{footnote}}\footnote[#1]{#2}\endgroup}

\usepackage[]{float,latexsym,times}

\usepackage{hyperref}
\usepackage[capitalise]{cleveref}

\graphicspath{{./img/}}

\theoremstyle{plain} 
\newtheorem{theorem}{Theorem}[section]
\newtheorem{lemma}{Lemma}[section]
\newtheorem{corollary}{Corollary}[section]
\newtheorem{problem}{Problem}

\theoremstyle{definition} 

\newtheorem{remark}{Remark}[section]

\crefname{corollary}{Corollary}{Corollaries}
\crefname{problem}{Problem}{Problems}
\crefname{remark}{Remark}{Remarks}
\crefname{appsec}{Appendix}{Appendices}
\crefname{lemma}{Lemma}{Lemmas}

\numberwithin{equation}{section} 
\numberwithin{remark}{section} 

\DeclareMathOperator{\E}{\mathbf{E}} 		
\DeclareMathOperator{\Var}{\mathbf{Var}} 	
\DeclareMathOperator*{\argmin}{\text{arg\,min\;}}

\newcommand{\R}{\mathbb{R}}
\newcommand{\nonNegSet}{\R_{\geq0}}

\newcommand{\Hyp}{\mathrm{H}} 	

\newcommand{\given}{\,|\,}
\newcommand{\bgiven}{\,\big|\,}
\newcommand{\bbgiven}{\,\bigg|\,}

\newcommand{\param}{\ensuremath{\theta}}
\newcommand{\paramRV}{\ensuremath{\Theta}}

\newcommand{\dd}{\mathrm{d}}	
\newcommand{\indd}[1]{\ensuremath{\mathbf{1}_{#1} }} 
\newcommand{\ind}[1]{\indd{\{#1\}} } 
\newcommand{\stat}[1]{\ensuremath{\mathbf{t}_{#1}}} 
\newcommand{\statWOa}{\ensuremath{{\stat{}}}} 

\newcommand{\errConstr}{\ensuremath{\kappa}} 
\newcommand{\updateProbMeasure}[1][n]{\ensuremath{Q_{\stat{#1}}}}
\newcommand{\updateProbMeasureTilde}[1][n]{\ensuremath{\tilde Q_{\stat{#1},\postProbwo_{#1}}}}
\newcommand{\updateProbMeasureTildeScaled}[2][n]{\ensuremath{\tilde Q_{\stat{#1},#2\postProbwo_{#1}^\bullet}}}
\newcommand{\updateProbMeasureMtx}[1][n]{\ensuremath{{\boldsymbol{Q}_{#1}}}}

\newcommand{\costC}{\ensuremath{d_n}}
\newcommand{\costS}{\ensuremath{g}}
\newcommand{\auxVarCost}{{\ensuremath{{D}}}}
\newcommand{\auxVarCostOpt}[1]{\ensuremath{{D^\star_{#1}}}}
\newcommand{\auxVarCostOptTilde}[1]{\ensuremath{{\tilde D^\star_{#1}}}}
\newcommand{\auxVarCostOptMtx}[1]{\ensuremath{{\boldsymbol{D}^\star_{#1}}}}
\newcommand{\auxVarCostOptDiscr}[1]{\ensuremath{{\boldsymbol{D}^\star_{#1}}}}

\newcommand{\RV}{\ensuremath{X}}
\newcommand{\RVidx}[1]{\ensuremath{\RV_{#1}}}
\newcommand{\SeqDataRV}{\ensuremath{\boldsymbol{X}_N}}
\newcommand{\obsIdx}[1]{\ensuremath{\boldsymbol{x}_{#1}}}
\newcommand{\obs}{\obsIdx{n}}
\newcommand{\obsScalar}[1]{x_{#1}}
\newcommand{\xnew}{\tilde{x}}	

\newcommand{\stopR}{\ensuremath{\Psi} }

\newcommand{\stopOptErr}{\ensuremath{\stopR^\star_{\errConstr}}}

\newcommand{\stopOptCerr}{\ensuremath{\stopR^\star_{C_\errConstr^\star}}}
\newcommand{\dec}{\ensuremath{\delta}}

\newcommand{\decOpt}{\ensuremath{\dec^\star}}

\newcommand{\stopAt}[1][]{\ensuremath{\Phi_{n#1}}}

\newcommand{\est}[1]{\ensuremath{\hat\param_{#1}}}

\newcommand{\policy}{\ensuremath{\pi}}
\newcommand{\policySet}{\ensuremath{\Pi}}
\newcommand{\policyOptC}{\ensuremath{\policy_C^\star}}
\newcommand{\policyOptCerr}{\ensuremath{\policy_{C_\errConstr^\star}^\star}}
\newcommand{\policyOptErr}{\ensuremath{\policy_{\errConstr}^\star}}

\newcommand{\policyFull}{\ensuremath{\{\stopR_n, \dec_n, \est{0,n},\est{1,n}\} _{0\leq n \leq N}}}
\newcommand{\policyMapping}{\ensuremath{\sampleSpaceObs[N]\rightarrow\{0,1\}^{N+1}\times\{0,1\}^{N+1}\times\parameterSpace_0^{N+1}\times\parameterSpace_1^{N+1}}}

\newcommand{\errorDet}[2]{\ensuremath{\alpha_{#1}^{#2}}}
\newcommand{\errorEst}[2]{\ensuremath{\beta_{#1}^{#2}}}
\newcommand{\errorDetTilde}[2]{\ensuremath{\tilde\alpha_{#1}^{#2}}}

\newcommand{\errorDetOpt}[2]{\ensuremath{{\alpha_{#1}^{#2}}}}
\newcommand{\errorEstOpt}[2]{\ensuremath{{\beta_{#1}^{#2}}}}

\newcommand{\stopRegionCompl}[1][n]{\ensuremath{\bar{\mathcal{S}}_{#1}}}
\newcommand{\stopRegionComplWoIdx}{\ensuremath{\bar{\mathcal{S}}}}
\newcommand{\stopRegion}[1][n]{\ensuremath{{\mathcal{S}_{#1}}}}
\newcommand{\stopRegionDec}[2]{\ensuremath{{\mathcal{S}_{#1}^{#2}}}}
\newcommand{\stopRegionDecWoIdx}[1]{\ensuremath{{\mathcal{S}^{#1}}}}
\newcommand{\stopRegionBound}[1][n]{\ensuremath{{\partial\mathcal{S}_{#1}}}}
\newcommand{\stopRegionBoundTilde}[1][n]{\ensuremath{{\partial\mathcal{\tilde S}_{#1}}}}

\newcommand{\Gam}{\mathrm{Gam}}
\newcommand{\unif}{\mathcal{U}}
\newcommand{\norm}[2]{\ensuremath{\mathcal{N}\left(#1,#2\right)}}

\newcommand{\metric}{\ensuremath{\mathcal{E}}}

\newcommand{\sampleSpaceObs}[1][]{\ensuremath{\Omega_X^{#1}}}
\newcommand{\parameterSpace}{\ensuremath{\Lambda}}
\newcommand{\sampleSpaceHyp}{\ensuremath{\Omega_\Hyp}}

\newcommand{\stateSpaceObs}[1][]{\ensuremath{E_X^{#1}}}
\newcommand{\stateSpaceParam}[1][]{\ensuremath{E_\parameterSpace^{#1}}}
\newcommand{\stateSpaceStat}[1][]{\ensuremath{E_\statWOa^{#1}}}
\newcommand{\stateSpaceHyp}[1][]{\ensuremath{E_\Hyp}}

\newcommand{\metricStat}{\ensuremath{\metric_\statWOa}}

\newcommand{\likelihoodRatio}{\ensuremath{\eta}}

\newcommand{\regConst}{\ensuremath{\varepsilon}}

\newcommand{\Nstat}{\ensuremath{N_t}}
\newcommand{\Nx}{\ensuremath{N_x}}

\newcommand{\postProbwo}{\ensuremath{e}}
\newcommand{\postProb}[2][n]{\ensuremath{\postProbwo_{#2,#1}}}

\newcommand{\transkernelPostVar}[3]{\ensuremath{\tilde\xi_{#3}(#1, #2)}}

\newcommand{\stateSpacePostProb}{\ensuremath{E_{\postProbwo}}}
\newcommand{\metricPostProb}{\ensuremath{\mathcal{E}_{\postProbwo}}}

\newlength{\imgWidthSingle}
\newlength{\imgWidthDouble}
  \acrodef{GLRT}{Generalized Likelihood Ratio Test}
 \acrodef{LP}{Linear Programming} 
 \acrodef{FSS}{fixed sample size}
 \acrodef{iid}{independent and identically distributed}
 \acrodef{KL}{Kullback-Leibler}
 \acrodef{MSE}{mean square error}
 \acrodef{MMSE}{minimum mean square error}
 \acrodef{SEL}{a}
 \acrodef{MAD}{sdfsf}
 \acrodef{LINEX}{linear exponential}
 \acrodef{pdf}{probability density function}
 \acrodef{SPRT}{Sequential Probability Ratio Test}
 \acrodef{w.r.t.}{with respect to}

\begin{document}
\title{\textbf{\Large Bayesian Sequential Joint Detection and Estimation}}

\date{}

\maketitle

\author{
	\begin{center}
		\vskip -1cm
		\textbf{\large Dominik Reinhard, Michael Fau\ss{}, and Abdelhak M. Zoubir} \\
		Signal Processing Group, Technische Universit\"at Darmstadt, \\
		Darmstadt, Germany
	\end{center}
}
\symbolfootnote[0]{\normalsize Address correspondence to Dominik Reinhard,
	Signal Processing Group, Technische Universit\"at Darmstadt, Merckstra\ss{}e 25, 
	64283 Darmstadt, Germany; E-mail: reinhard@spg.tu-darmstadt.de
}

{\small \noindent\textbf{Abstract:}
Joint detection and estimation refers to deciding between two or more hypotheses and, depending on the test outcome, simultaneously estimating the unknown parameters of the underlying distribution.
This problem is investigated in a sequential framework under mild assumptions on the underlying random process.
We formulate an unconstrained sequential decision problem, whose cost function is the weighted sum of the expected run-length and the detection/estimation errors.
Then, a strong connection between the derivatives of the cost function with respect to the weights, which can be interpreted as Lagrange multipliers, and the detection/estimation errors of the underlying scheme is shown.
This property is used to characterize the solution of a closely related sequential decision problem, whose objective function is the expected run-length under constraints on the average detection/estimation errors.
We show that the solution of the constrained problem coincides with the solution of the unconstrained problem with suitably chosen weights.
These weights are characterized as the solution of a linear program, which can be solved using efficient off-the-shelf solvers.
The theoretical results are illustrated with two example problems, for which optimal sequential schemes are designed numerically and whose performance is validated via Monte Carlo simulations.  }

\vspace*{1em}
\\[1em]
{\small \noindent\textbf{Subject Classifications:} 
	62L10; 62L12; 62L15; 90C05; 93E10.

\section{Introduction}

In a wide range of applications, detection and estimation appear in a coupled way and both are of primary interest. This means that one wants to decide between two or more hypotheses and, depending on the decision, estimate one or more, possibly random, parameters of the underlying distribution. This problem was initially treated by \citet{Middleton1968Simultaneous}, who used a Bayesian framework to obtain a jointly optimal solution. After the extension to a framework for testing multiple hypotheses by \citet{fredriksen1972simultaneous}, the joint detection and estimation problem received little to no attention in the literature. However, more recently, joint detection and estimation has regained importance \citep{tajer2010Optimal, moustakides2012joint, Moustakides2011Optimal, Momeni2015Joint,Li2007Optimal,Li2016Optimal}.

In speech processing, for example, one is interested in detecting whether a speech signal is present or not and in the former case, estimating the speech spectral amplitude \citep{Momeni2015Joint}. If one wants to use dynamic spectrum access in a cognitive radio, the secondary user has to detect the primary user and has to estimate the possible interference \citep{Yilmaz2014Sequential}. There exist many other areas such as change point detection and estimation of time of change \citep{boutoille2010hybrid}, radar \citep{tajer2010Optimal}, optical communications \citep{wei2018simultaneous}, detection and estimation of objects from images \citep{Vo2010Joint} or biomedicine \citep{jajamovich2012Minimax, chaari2013fast} to name just a few.

For all these applications, detection and estimation are intrinsically coupled and a reliable decision and accurate estimates are of primary interest. Treating  both problems separately and finding an optimal solution for each of them does not result in a jointly optimal solution \citep{Moustakides2011Optimal}. An example would be to use the Neyman-Pearson test as an optimal solution for the detection part and a Bayesian estimator for the estimation task.
There exist other approaches, such as the \ac{GLRT}, which solve the problem in a combined manner but only the detection performance is of primary interest.

Sequential analysis is a field of research which was introduced by \citet{wald1947sequential} with his famous \ac{SPRT}. The idea behind the \ac{SPRT} is to design a test which uses as few samples as possible while guaranteeing that constraints on the two error probabilities are fulfilled. Since the introduction of the \ac{SPRT}, sequential detection and estimation methods have been developed \citep{ghosh2011sequential,tartakovsky2014sequential}. Especially for time critical or low energy applications, sequential methods are preferable to fixed sample size ones.

Combining the idea of sequential analysis with joint detection and estimation leads to a powerful framework that uses as few sample as possible while fulfilling constraints on the error probabilities as well as on the quality of the estimates.

Although there exist sequential tests which deal with composite hypotheses and that include an estimation part \citep{li2014generalized, goelz2017}, estimation is usually of secondary interest. The joint solution of sequential detection and estimation has attained little attention in the literature. \citet{Yilmaz2014Sequential} addressed the problem of sequential joint spectrum sensing and channel estimation. The aim of that work is to estimate a communication channel and to maximize the secondary user throughput while a constraint on the primary user outage is fulfilled. That approach fulfills a constraint on the overall target accuracy and minimizes the expected number of samples.
The question, whether the primary user is active or not is formulated as a hypothesis test and, hence, the authors end up with a sequential joint detection and estimation formulation.

In \citet{reinhard2016}, we proposed an approach which treats the sequential joint detection and estimation problem in a non-Bayesian framework. The goal is to perform a sequential hypothesis test and estimate an unknown quantity if one decides in favor of the alternative where the underlying distributions are not known exactly.

Moreover, \citet{Yilmaz2016Sequential} provided an optimal sequential joint detection and estimation framework for multiple hypotheses which is based on a state space model. That approach uses, similarly to \citep{Yilmaz2014Sequential}, an overall cost function consisting of a weighted combination of detection and estimation errors. The run-length is then minimized such that the cost function fulfills a certain constraint for every set of observations. Although that scheme is optimal for fixed weighting coefficients, the question how to chose these coefficients when one wants to achieve a certain performance in terms of error probabilities or estimation error remains unanswered. Especially since the error probabilities and the estimation errors have very different numerical ranges, a trade-off between detection and estimation errors is very hard and choosing the coefficients heuristically is rather impossible.

Recently, \citet{fauss2017sequential} proposed an approach which tests the presence or absence of a signal in a sequential setting and in case the signal is present, the signal-to-noise ratio is estimated. 

Contrary to existing methods, this work provides a strictly optimal sequential joint detection and estimation framework under mild assumptions. The proposed approach is based on a Bayesian framework and minimizes the average posterior risk. The loss function jointly takes detection and estimation errors into account as well as the expected number of samples. In addition to this, a method is provided to obtain the optimal coefficients of the cost function such that the resulting scheme fulfills constraints on the detection performance as well as on the estimation quality while the expected run-length is minimized.

This paper starts with a statement of the underlying assumptions, followed by a short explanation of the basics of Bayesian decision theory and optimal stopping theory with application to sequential joint detection and estimation. The resulting problem is then reduced to an optimal stopping problem whose solution is characterized by a non-linear Bellman equation. The properties of this equation are presented and a connection between the Bellman equation and the estimation and detection errors is derived.
Exploiting this connection enables us to obtain the coefficients of the loss function such that the resulting test is of minimum run-length and fulfills predefined constraints on the detection and estimation errors.
The advantages of the presented method and special cases are highlighted in the discussion section. Numerical results are then provided to illustrate the behavior of the presented scheme.
 \section{Problem Statement}\label{sec:problemStatement}
Let $\SeqDataRV = (\RVidx{1},\ldots,\RVidx{N})$ be a sequence of random variables. The sequence $\SeqDataRV$ can be generated under two different hypotheses $\Hyp_i$, $i=0,1$. Under each hypothesis, the distribution of the random variables $\RVidx{n}$, $n=1,\ldots,N$, depends on a random parameter $\paramRV$ whose distribution is determined by the underlying hypothesis $\Hyp_i$.
Hence, the tuple $(\RVidx{n}, \paramRV, \Hyp_i)$ in the metric state space $(\stateSpaceObs \times E_\paramRV \times \stateSpaceHyp, \mathcal{E}_\RV \times \mathcal{E}_\paramRV \times \mathcal{E}_\Hyp)$ is defined on the probability space $\left(\sampleSpaceObs \times \parameterSpace \times \sampleSpaceHyp, \mathcal{F}_\RV \otimes \mathcal{F}_{\paramRV} \otimes \mathcal{F}_{\Hyp}, P\right)$. Furthermore, the random variables $\RVidx{n}\given\paramRV$ are conditionally \ac{iid}.
The joint probability density of the tuple $(\SeqDataRV, \paramRV, \Hyp_i)$ can hence be factorized as
\begin{align*}
 p(\Hyp_i,\theta,\obsIdx{N}) = p(\obsIdx{N}\given\theta) p(\theta\given\Hyp_i)p(\Hyp_i) = \prod_{n=1}^N p(\obsScalar{n}\given\theta) p(\theta\given\Hyp_i)p(\Hyp_i) \,,
\end{align*}
where $\obsIdx{n}=(x_1,\ldots,x_n)$ and $\param$ are the realizations of $\SeqDataRV$ and $\paramRV$, respectively.
The two composite hypotheses can be formulated as:
\begin{align}
 \begin{split}
  \Hyp_0: & \quad \SeqDataRV\given\paramRV_0 \sim p(\obsIdx{N}\given\param_0),\;\paramRV_0\sim p(\param_0\given\Hyp_0)\\
  \Hyp_1: & \quad \SeqDataRV\given\paramRV_1 \sim p(\obsIdx{N}\given\param_1),\;\paramRV_1\sim p(\param_1\given\Hyp_1)
  \label{eq:problemSatement}
 \end{split}
\end{align}
Furthermore, let $\parameterSpace_i$, $i=0,1$, denote the parameter space of $\paramRV_i$ and let $\parameterSpace$ = $\parameterSpace_0 \cup \parameterSpace_1$. The random parameters $\paramRV_i$ as well as the parameter spaces $\parameterSpace_i$ can differ under the hypotheses $\Hyp_i$.
Since the two hypotheses are composite, one is also interested in estimating the underlying parameter $\paramRV_i$ which then results in a joint detection and estimation problem. Since the sequence $\SeqDataRV$ is observed sequentially, one obtains a sequential joint detection and estimation problem. Usually one is interested in a procedure which guarantees that the error probabilities as well as the quality of the estimates are not worse than the pre-specified levels.

Put simply, the problem can be formulated as follows: 
\emph{design a sequential procedure which uses on average as few samples as possible and fulfills constraints on the error probabilities and the estimation quality.}

Before a more technical problem formulation is given, the underlying assumptions are summarized and some remarks on the notation are given.

\subsection{Assumptions and Notation}
The following assumptions are made throughout the paper.
\begin{enumerate}
 \item The random variables $X_n$, $n=1,\ldots,N$, are conditionally \ac{iid} with respect to $\paramRV$, that is
	\begin{align*}
	  p(\obsIdx{N}\given\param) = &\prod_{n=1}^Np(\obsScalar{n}\given\param)\,, \\
	  p(\obsIdx{N}\given\param,\Hyp_i) = &\prod_{n=1}^Np(\obsScalar{n}\given\param,\Hyp_i)\,.
	\end{align*}
  \item The hypothesis $\Hyp_i$, $i=0,1$, as well as the random parameter $\paramRV_i$ do not change during the observation period of \SeqDataRV.
  \item\label{ass:suffStat} A sufficient static $\stat{n}(\obs)$ in a state space $(\stateSpaceStat,\metricStat)$ exists such that
	\begin{align*}
	  p(\param\given\obs) = &p(\param\given\stat{n}(\obs))\,,\quad \forall n=1,\ldots,N\,.
	\end{align*}
	 The sufficient statistic has a transition kernel of the form
	\begin{align*}
	 \stat{n+1}(\obsIdx{n+1}) = \xi(\stat{n}(\obs),\obsScalar{n+1}) =: \xi_{\stat{n}}(\obsScalar{n+1})\,.
	\end{align*}
	and some initial statistic $\stat{0}$.
  \item The two hypotheses are separable, i.e.,
	\begin{align*}
	 \exists i \in\{0,1\}: \lim_{N\rightarrow\infty} p(\Hyp_i\given\obsIdx{N}) = 1\,.
	\end{align*}
  \item \label{ass:finiteMoments}The second order moment of the random parameter $\paramRV$ exists and is finite, i.e.,
	\begin{align*}
	 \E[\paramRV^2] < &\infty\,.\\
	\end{align*}
	Note that this implies, that the conditional second order moment $\E[\paramRV^2\given A]$ exist and is finite for all events $A$ with non-zero probability.

\end{enumerate}
For the ease of notation, the dependency of functions, like estimators, on the observations is not mentioned explicitly, but should be clear from the context. Moreover, the time index $n$ for a sequence of observations is dropped in the case the number of variables is not important for the statement.
In case integrals are taken over the entire domain, the integration domain is left out for a shorter notation, i.e.,
\begin{align*}
 \E\bigl[X\bigr] = \int_{\stateSpaceObs} xp(x)\dd x = \int x p(x)\dd x\,.
\end{align*}

\subsection{Formal Problem Definition}
Most statistical inference tasks, like detection and estimation, can be treated as decision making tasks, see, e.g., \citet{berger1985statistical}. This decision theoretic framework is used to tackle the general joint detection and estimation problem. The relevant concepts of this framework are introduced by means of the pure detection case. This is followed up by the pure estimation case which then leads to a formulation of the joint detection estimation problem. This formulation is then transferred to the sequential case to obtain the final problem formulation.

\subsubsection*{Detection}
In hypothesis testing, we introduce a loss function $L(a,h)$, which assigns a cost to each action $a$ if the state $h$ is the true state. In the case of a detection problem the action $a$ is the decision in favor of one hypothesis and the state $h$ is the true hypothesis. The loss function considered here is the generalized $0/1$ loss function:
   \begin{align}
    L^\text{D}(a,h) = & \left\{
		     \begin{array}{ll}
                     C_h & \text{if } a \neq h \\
                     0  & \text{if } a = h	      
		     \end{array}
		     \right.
   \end{align}
To obtain the policy, i.e., a mapping from any possible sequence of observations to a particular action, the \emph{posterior expected loss} has to be introduced. In a detection context, the policy is also called the decision rule and is denoted by $\dec$ from now on. The posterior expected loss is the expected value of the loss function conditioned on the observations, i.e.,
\begin{align*}
 R^\text{D}(\dec) = \E_\Hyp[ L^\text{D}(\dec,\Hyp) \given \obsIdx{}]\,
\end{align*}
where $\E_\Hyp[\cdot]$ denotes the expected value with respect to the random variable $\Hyp$. For a given sequence of observations $\obsIdx{}$, the optimal decision rule $\decOpt$ is the one which minimizes the expected posterior risk.

\subsubsection*{Estimation}
Similar to detection, we introduce a loss function $L(a,\param)$, which assigns a cost to an action $a$, i.e., an estimate, when the parameter $\param$ is true.
For estimation problems, there exist a variety of different loss functions, e.g., the \textit{squared-error loss} or the \textit{absolute-error loss}. In this work, only the \textit{squared-error loss} is used, which is defined as
\begin{align}
 L^\text{E}(a,\param) = (a-\param)^2\,.
\end{align}
As in detection, the posterior expected loss is considered to find the policy, also called estimator $\hat\param$. The posterior expected loss is given by
\begin{align*}
 R^\text{E}(\hat\param) = \E_\paramRV[L^\text{E}(\hat\param,\paramRV)\given\obsIdx{}]\,.
\end{align*}
The optimal estimator, which minimizes the expected posterior loss when using the squared-error loss, is the posterior mean \citep{levy2008principles}, i.e.,
\begin{align}
 \est{}^\star = \E\left[\paramRV\given \obsIdx{}\right]\,. 
\end{align}
The corresponding posterior expected risk is then
\begin{align}
 R^\text{E}(\est{}^\star) = \Var\left[\paramRV\given \obsIdx{}\right]\,.
\end{align}

\subsubsection*{Joint Detection and Estimation}
As opposed to the pure detection or pure estimation case, in which a single action and a single state is considered, joint detection and estimation is more complex. The state of nature of this problem is the hypothesis $\Hyp_i$ as well as the random parameter $\paramRV$. The action is a tuple consisting of a decision, corresponding to the detection part, and two estimates depending on which hypothesis was chosen. The decision is denoted by $d$ and the two estimates are denoted by $e_0$ and $e_1$ for $\Hyp_0$ and $\Hyp_1$, respectively. The proposed loss function is then given by:

\begin{align} \label{eq:lossFctJoint}
    L^\text{J}(d, e_0, e_1, h, \param) = \left\{
    \begin{array}{ll}
      C_h  & \text{for } d\neq h \\
      C_{h+2} (\param-e_{h})^2 & \text{for } d = h \\
    \end{array}    
    \right.
\end{align}
Literally speaking, this loss functions means that a fixed cost is assigned for a wrong decision, whereas the weighted squared-error of the estimate is assigned in the case of a correct decision.
Contrary to other authors, e.g., \citet{Yilmaz2016Sequential}, we do not incorporate any detection error in the case of a wrong decision.

To obtain the policy, which consists of a decision rule and two estimators, we define the posterior expected loss as:
\begin{align}\label{eq:RjointFSS}
 R^\text{J}(\dec,\est{0},\est{1}) & = \E_{\paramRV,\Hyp}\biggl[L^\mathrm{J}(\dec,\est{0},\est{1},\Hyp,\paramRV)\given\obsIdx{}\biggr] \\
		    & =\sum_{i=0}^1p(\Hyp_i\given\obsIdx{}) \int L^\mathrm{J}(\dec,\est{0},\est{1},\Hyp_i,\param)p(\param\given\Hyp_i,\obsIdx{}) \dd\param
\end{align}
Again, the optimal policy is the policy which minimizes \cref{eq:RjointFSS} for any possible sequence of observations.

To find an optimal policy for the sequential joint detection and estimation problem, the previous results have to be included in an optimal stopping framework. Optimal stopping is important when designing strictly optimal sequential schemes---see for example \citet{novikov2009optimal} for sequential detection and \citet{ghosh2011sequential} for sequential estimation. The aim is to find a \emph{stopping rule} which maps the observation space onto a decision whether to stop or to continue sampling. Mathematically, a non-randomized stopping rule $\stopR=\left(\stopR_n\right)_{n\geq0}$ is defined by
\begin{align*}
 \stopR_n: \sampleSpaceObs[n] \rightarrow\{0,1\}, \quad n\in\mathbb{N}_0\;.
\end{align*}
A randomized stopping rule would instead map the observation space to a probability whether to stop or to continue sampling. Only non-randomized stopping rules are considered in this work. The reason is stated in \cref{lemma:boundaryNullSet}.

The stopping time or run-length $\tau$ of a sequential scheme is defined as
\begin{align}
 \tau = \inf\left\{n: \stopR_n=1\right\}\,.
\end{align}
When designing sequential schemes, the stopping rule should trade-off the expected run-length and the accuracy of the sequential scheme, e.g., the error probabilities in sequential detection.

For the sake of a more compact notation, the auxiliary variable  
\begin{align*}
 \stopAt := \stopR_n\prod_{i=0}^{n-1}(1-\stopR_i) = \ind{\tau=n}
\end{align*}
is introduced, where $\ind{\mathcal{A}}$ denotes the indicator function of event $\mathcal{A}$.

There exist two types of optimal stopping problems: the \emph{infinite horizon} problem and the \emph{finite horizon} problem, also referred as truncated problem. In the truncated case, the sequential procedure is forced to stop at a given time $N$, i.e., $\stopR_N=1$. In this work, only finite horizon problems are considered.
An overview of optimal stopping techniques and its applications can be found in e.g. \citet{peskir2006optimal}.

In the sequential case, where the number of used samples is not known beforehand, the overall posterior expected risk becomes
\begin{align}\label{eq:errCostStop}
    \sum_{n=0}^N \stopAt R^\text{J}_n(\dec_n,\est{0,n}, \est{1,n})\,,
\end{align}
where $R^\text{J}_n$ is similar to the one defined in \cref{eq:RjointFSS}, except that the data, as well as the estimators and decision rule now depend on the time index $n$.

Having introduced the concept of optimal stopping and the basics of Bayesian decision theory, the problem of sequential joint detection and estimation can be formulated in a formal manner. The tuple of the stopping rule, the decision rule and the estimators under both hypotheses is referred to policy in the following sections and is defined by
\begin{align}\label{eq:policyDef}
 \policy = \policyFull\,.
\end{align}
Moreover, the set of feasible policies is defined by
\begin{align}
 \policySet: \policyMapping\,.
\end{align}
To find an optimal policy, we define the objective function as the sum of the expected run-length and the expected value of the costs stated in \cref{eq:errCostStop}, i.e.,
\begin{align}
 J(\policy) = & \E[\tau] +  \E\biggl[ \sum_{n=0}^N \stopAt R^\text{J}_n(\dec_n,\est{0,n}, \est{1,n}) \biggr]\,,\\
 = & \E\left[\sum_{n=0}^N \stopAt \left( n + R_n^\text{J} (\dec_n,\est{0,n}, \est{1,n}) \right)\right]\,.
\end{align}
This objective function combines the costs obtained from the Bayesian decision problem with the one of the optimal stopping problem and is used to find the optimal policy. This is stated in the following problem.

\begin{problem} \label{prb:genProblem} For given and finite $C_i\geq0$, $i\in\{0,1,2,3\}$, and fixed $N\geq1$, the optimal policy $\policyOptC$ with respect to the costs $C_i$ is found by solving
\begin{align}
\min_{\policy\in\policySet}\;\E\left[\sum_{n=0}^N \stopAt \left( n + R_n^\text{J} (\dec_n,\est{0,n}, \est{1,n}) \right)\right]\,.
\end{align}
\end{problem}
This problem formulation is used below obtain an optimal policy. The choice of $C_i$ and the influence of $N$ is discussed in \cref{sec:optCostCoeff,sec:discussion}, respectively.

\subsection{Performance Measures for Sequential Joint Detection and Estimation}\label{subsec:performanceMeasures}
Before a solution of \cref{prb:genProblem} is presented, the question how to analyze the performance of a joint detection and estimation procedure is discussed.
One common quality measure for sequential schemes is the expected run-length \citep{ghosh2011sequential,wald1947sequential}, which usually corresponds to the expected number of samples used by a sequential procedure. Similarly to a pure sequential detection problem, the detection performance can be quantified by the type I and type II error probabilities.
For the estimation part, the same quality measure as in the combined loss function, i.e., the \ac{MSE}, is used. The error measures can be written as
\begin{align}
  \errorDet{n}{i}(\stat{n}) &= \E[\ind{\dec_\tau=1-i}\given \stat{n}, \Hyp_i, \tau\geq n]\,, \quad \, i\in\{0,1\}\,,\quad 0\leq n \leq N\,,\label{eq:alphaNonRec}\\
  \errorEst{n}{i}(\stat{n}) &= \E[\ind{\dec_\tau=i}(\est{i,\tau} - \theta)^2\given \stat{n},\Hyp_i, \tau\geq n]\,, \quad \, i\in\{0,1\}\,,\quad 0\leq n \leq N\,. \label{eq:betaNonRec}
\end{align}
The quantity $\errorDet{n}{i}(\stat{n})$ denotes the probability that the sequential scheme makes an erroneous decision in favor of $\Hyp_{1-i}$ given that the hypothesis $\Hyp_i$ is true and the test is in state \stat{n} at time instant $n$. Analogously, $\errorEst{n}{i}(\stat{n})$ is the \ac{MSE} of the sequential scheme under the hypothesis $\Hyp_i$ given that a correct decision is made and given the statistic \stat{n} at time instant $n$.
The quantities defined in \cref{eq:alphaNonRec,eq:betaNonRec} can be defined in the following recursive manner, also known as Chapman-Kolmogorov backward equations \citep{EoM1990}, 
\begin{align}
  \begin{split}
 \label{eq:alphaRec}
  \errorDet{n}{i}(\stat{n}) &= \ind{\dec_n=1-i}\stopR_n  + (1-\stopR_n)\E[ \errorDet{n+1}{i}(\stat{n+1}) \given \stat{n}, \Hyp_i]\,, \quad 0\leq n < N\,,\\
  \errorDet{N}{i}(\stat{N}) &= \ind{\dec_N=1-i}\stopR_N\,,
  \end{split}\\
  \begin{split}
 \label{eq:betaRec}
    \errorEst{n}{i}(\stat{n}) &=  \E[(\est{i,n}-\theta)^2\given \stat{n}, \Hyp_i ]\ind{\dec_n=i}\stopR_n  +  (1-\stopR_n)\E[ \errorEst{n+1}{i}(\stat{n+1})\given \stat{n},\Hyp_i]\,, \quad 0\leq n < N\,,\\
    \errorEst{N}{i}(\stat{N}) &=  \E[(\est{i,N}-\theta)^2\given \stat{N}, \Hyp_i ]\ind{\dec_N=i}\stopR_N\,.
    \end{split}
\end{align}
The importance of the recursive definition becomes clear later on.

After the definitions of the performance measures were given, the optimal sequential joint detection and estimation scheme is introduced in the next sections.
\cref{prb:genProblem} is first reduced to an optimal stopping problem which is shown in the next section.
The properties of the resulting cost function, which are essential for the proof of optimality, are discussed in \cref{sec:propCostFct}.
Finally, an approach for selecting the coefficients of the cost function such that the optimal procedure is guaranteed to meet pre-specified performance constraints is presented in \cref{sec:optCostCoeff}.
 
\section{Reduction to an Optimal Stopping Problem}
To obtain an optimal solution of \cref{prb:genProblem}, we first need to solve it with respect to the decision rule and the estimators so as to end up with an optimal stopping problem.
The first term of the objective function in \cref{prb:genProblem}, i.e., the expected run-length, can be written as
\begin{align} \label{eq:jointStoppingTime}
 \E[\tau] & = \sum_{n=0}^N \sum_{i=0}^{1}\int \int n \stopAt p(\obs,\theta,\Hyp_i) \dd\theta \dd\obs \nonumber\\
	    & = \sum_{n=0}^N \int n \stopAt p(\obs) \dd\obs\;.
\end{align}
In a similar way, the second part of the objective function can be expressed as
\begin{align} \label{eq:jointRiskExpl}
\E\biggl[\sum_{n=0}^N \stopAt R^\text{J}_n\biggr] = & \sum_{n=0}^N \int \sum_{i\in\{0,1\}} \stopAt \ind{\dec_n=i}\auxVarCost_{i,n} p(\obs)\dd\obs\,,
\end{align}
with the short hand notation
\begin{align} 
	\auxVarCost_{i,n} = & \int C_{1-i}p(\Hyp_{1-i}\given\obs)p(\theta\given\Hyp_{1-i},\obs) + C_{2+i}p(\Hyp_i\given\obs)(\est{i,n}-\param)^2 p(\theta\given\Hyp_i,\obs)\dd\theta  \\
      \label{eq:dTildeLast}
      = & C_{1-i}p(\Hyp_{1-i}\given\obs) +  C_{2+i}p(\Hyp_i\given\obs) \int (\est{i,n}-\param)^2 p(\theta\given\Hyp_i,\obs)\dd\theta \,.
\end{align}
The auxiliary variable $\auxVarCost_{i,n}$ can be interpreted as the combined detection and estimation cost if the test stops at time $n$ and decides in favor of $\Hyp_i$. 
According to Assumption \ref{ass:suffStat}, the conditioning on the data itself in \cref{eq:dTildeLast} can be replaced by conditioning on the sufficient statistic $\stat{n}$.

The overall objective, which should be minimized by the optimal policy, is then a weighted sum of the expected run-length, the error probabilities and the \ac{MSE}, i.e.,
\begin{align}
 \label{eq:objWeightedSum}
 \begin{split}
 \E[\tau] & + C_0p(\Hyp_0)p(\dec_\tau=1\given\Hyp_0) + C_1p(\Hyp_1)p(\dec_\tau=0\given\Hyp_1) \\
 & + C_2p(\Hyp_0)\E[\ind{\dec_\tau=0} (\est{0,\tau} - \theta)^2\given\Hyp_0] + C_3p(\Hyp_1)\E[\ind{\dec_\tau=1} (\est{1,\tau} - \theta)^2\given\Hyp_1]\,.
 \end{split}
\end{align}

As one can see, \cref{eq:jointStoppingTime} depends on the stopping rule only, whereas \cref{eq:jointRiskExpl} depends on the overall policy.
Hence, \cref{eq:jointRiskExpl} is first minimized \ac{w.r.t.} the detection rule and the estimators $\est{0,n}$ and $\est{1,n}$, $n=1,\ldots,N$.

Following the usual line of arguments for optimal sequential tests---see, e.g., \citep[Theorem 2.2]{novikov2009optimal}---it holds that
\begin{align}\label{eq:ineqDetRule}
  & \sum_{n=0}^N \int \sum_{i\in\{0,1\}} \stopAt \ind{\dec_n=i}\auxVarCost_{i,n} p(\obs) \dd\obs \geq 
   \sum_{n=0}^N \int \stopAt \min\left\{\auxVarCost_{0,n},\auxVarCost_{1,n}\right\} p(\obs) \dd\obs\,,   
\end{align}
where equality holds if and only if
\begin{align}\label{eq:optDecRule}
 \ind{\auxVarCost_{0,n} > \auxVarCost_{1,n}} \leq \dec^\star_n \leq & \ind{\auxVarCost_{0,n} \geq \auxVarCost_{1,n}}\,.
\end{align}
The short hand notation
\begin{align} \label{eq:shortHandSetIneq}
    \{\auxVarCost_{0,n} > \auxVarCost_{1,n}\} := \{\obs\in\stateSpaceObs[n]:\auxVarCost_{0,n}(\obs) > \auxVarCost_{1,n}(\obs)\}\,.                                                                           
\end{align}
is used throughout the paper.

\begin{remark}\label{rmk:nonRndDecR}
The choice whether to decide systematically for one specific hypothesis if both are equally costly does not affect the solution of the optimal stopping problem. Only when one tries to chose the coefficients $C_i$, $i\in\{0,1,2,3\}$, such that conditions on the error probabilities are fulfilled, this case plays a small but noteworthy role. Generally, one can say that the set in the state space for which both hypotheses are equally costly, is always part of the complement of the stopping region for $n<N$. For $n=N$, the test has to decide between the two hypotheses even if they are equal costly but since we assume a sufficiently large $N$, i.e., $P(\tau=N)\approx0$, also a systematic decision in favor of one hypothesis would hardly affect the performance of the test. See \cref{sec:numResults} for details.
\end{remark}

Before the optimal stopping problem can be solved, the right hand side of \cref{eq:ineqDetRule} has to be optimized \ac{w.r.t.} to the estimators $\est{i,n}$, $i=0,1$, $n=1,\ldots,N$. 

According to Gosh et al., the optimal sequential estimator is independent of the stopping rule \citep[Theorem 5.2.1]{ghosh2011sequential}. As a consequence of this and due to the fact that each $\auxVarCost_{i,n}$ only depends on one estimator, it is sufficient to minimize both $\auxVarCost_{i,n}$, $i=0,1$,  \ac{w.r.t.} the corresponding estimators separately. The optimal estimators, i.e., the \ac{MMSE} estimators,  are then given by
\begin{align}\label{eq:optSeqEst}
\est{i,n}^\star = \E\bigl[\paramRV\given\obs,\Hyp_i\bigr]\,.
\end{align}
Using the optimal estimators and the optimal decision rules, the right hand side of \cref{eq:ineqDetRule} can be expressed as
\begin{align*}
 \sum_{n=0}^N \int\stopAt \min\left\{\auxVarCostOpt{0,n},\auxVarCostOpt{1,n}\right\}p(\obs)\dd\obs   = & \sum_{n=0}^N \E\left[\stopAt \min\left\{\auxVarCostOpt{0,n},\auxVarCostOpt{1,n}\right\}\right] \,,
\end{align*}
where the combined detection/estimation cost \auxVarCostOpt{i,n} for deciding in favor of $\Hyp_i$ at time $n$ is given by
\begin{align}
      \auxVarCostOpt{i,n} & = C_{1-i}p(\Hyp_{1-i}\given\obs) + C_{2+i}p(\Hyp_i\given \obs)\Var\bigl[\paramRV\given\obs,\Hyp_i\bigr] \nonumber\\ 
	& = C_{1-i}p(\Hyp_{1-i}\given\stat{n}) \, + C_{2+i}p(\Hyp_i\given \stat{n})\,\Var\bigl[\paramRV\given\stat{n},\Hyp_i\bigr]\label{eq:defAuxCostOpt}\,. 
\end{align}
The combined detection/estimation  cost $\auxVarCostOpt{i,n}$ for deciding in favor of $\Hyp_i$ at time $n$ is a linear combination of the detection and the estimation costs.
More precisely, the first part is the posterior probability that $\Hyp_{1-i}$ is true, whereas the second part is the posterior variance of the random parameter under $\Hyp_i$ weighted by the posterior probability that $\Hyp_i$ is true.
From now on, we replace the data in all conditional probabilities and expectations by the sufficient statistic $\stat{n}$. Opposed to the data, the sufficient statistic is of fixed dimension according to Assumption \ref{ass:suffStat}. The short hand notation stated in \cref{eq:shortHandSetIneq} is then defined on $\stateSpaceStat$ instead on $\stateSpaceObs[n]$

Using the previous results, \cref{prb:genProblem} reduces to the following optimal stopping problem
\begin{align}\label{eq:optStoppingFormulation}
 \min_\stopR\;\sum_{n=0}^N\E\left[ \stopAt\left(n+\costS(\stat{n})\right)\right]\,,
\end{align}
with the instantaneous cost
\begin{align}\label{eq:gCostFct}
 \costS(\stat{n}) = \min\left\{\auxVarCostOpt{0,n},\auxVarCostOpt{1,n}\right\}\,.
\end{align}
The solution of the optimal stopping problem is stated in the next theorem.
\begin{theorem}\label{theo:solutionOptimalStopping}
 The solution of \cref{eq:optStoppingFormulation} can be characterized by the system of functional equations
 \begin{align}
    \begin{split}
      \rho_n(\stat{n}) = & \min\biggl\{ \costS(\stat{n}), \costC(\stat{n}) \biggr\}\,, \quad n < N\,, \\
      \rho_N(\stat{N}) = & \costS(\stat{N})\,,
    \end{split}
 \end{align}
 with the short-hand notation
 \begin{align}\label{eq:updateComplicated}
  \costC(\stat{n}) = & 1+ \int \rho_{n+1}\left(\xi(\stat{n},\xnew)\right) p(\xnew\given\stat{n}) \dd \xnew 
 \end{align}
and $\costS(\stat{n})$ as defined in \cref{eq:gCostFct}. 
\end{theorem}
A proof of \cref{theo:solutionOptimalStopping} is given in \cref{app:proofSolutionOptimalStopping}.
With a change in measure, \cref{eq:updateComplicated} can be written as
\begin{align}
 \begin{split}
  \costC(\stat{n}) = & 1+ \int \rho_{n+1}\dd\updateProbMeasure\,,
 \end{split}
\end{align}
where
\begin{align}\label{eq:updateProbMeasure}
 \updateProbMeasure(B)  := P\left(\left\{\xnew\in \stateSpaceObs: \xi(\stat{n},\xnew)\in B \right\} \given \stat{n} \right)\,,
\end{align}
for all elements $B$ of the Borel $\sigma$-algebra on $E_\statWOa$. The probability measure $\updateProbMeasure^i$, for $i=0,1$, uses the probability measure $P(\cdot\given \stat{n}, \Hyp_i)$ instead of $P(\cdot\given\stat{n})$.
\begin{corollary}\label{corr:H-integrable}
The function $\rho_{n+1}(\stat{n+1})$ is \updateProbMeasure-integrable for all $\stat{n}$ and all $0\leq n < N$. 
\end{corollary}
A proof of \cref{corr:H-integrable} is given in \cref{app:proof_corr_H-integrable}.
The results on the optimal policy are summarized in the following corollary.
\begin{corollary}\label{cor:optPolicyC}
 The optimal policy $\policyOptC$ which solves \cref{prb:genProblem} is given by
 \begin{align*}
   \policyOptC = \{\stopR^\star_n,\dec^\star_n,\est{0,n}^\star,\est{1,n}^\star\}_{0\leq n \leq N}
 \end{align*}
 with $\dec^\star_n$ defined in \cref{eq:optDecRule} and $\est{i,n}^\star$ defined in \cref{eq:optSeqEst}. The optimal stopping rule $\stopR^\star_n$ is given by
 \begin{align}
  \stopR^\star_{n} = &\ind{\rho_n(\stat{n}) = \costS(\stat{n})}\,.
 \end{align}
\end{corollary}
The optimal stopping rule $\stopR^\star_n$ follows directly from the definition of $\rho_n$, see, e.g., \citep{novikov2009optimal,poor2009quickest}. 

The stopping region of the test, its boundary and its complement using the optimal policy from \cref{cor:optPolicyC} are defined as
\begin{align}\label{eq:defRegionsn}
 \begin{split}
  \stopRegion = & \left\{\stat{n}\in E_\statWOa: g(\stat{n}) < 1+ \int \rho_{n+1}\dd\updateProbMeasure\right\} \\
  \stopRegionBound = & \left\{\stat{n}\in E_\statWOa: g(\stat{n}) = 1+ \int \rho_{n+1}\dd\updateProbMeasure\right\} \\
  \stopRegionCompl = & \left\{\stat{n}\in E_\statWOa: g(\stat{n}) > 1+ \int \rho_{n+1}\dd\updateProbMeasure\right\} \\
 \end{split}
\end{align}
for $n<N$. Note that due to the truncation of the sequential scheme the stopping region and its complement for $n=N$ are defined as:
\begin{align}\label{eq:defRegionsN}
  \begin{split}
    \stopRegion[N] = & E_\statWOa\\
    \stopRegionCompl[N] = & \emptyset 
  \end{split}
\end{align}
 \section{Properties of the Cost Function}\label{sec:propCostFct}
In this section, some fundamental properties of the cost function are shown. These properties are later used to obtain a set of optimal coefficients of the cost function.

In order to simplify the derivations, it is useful to show that the boundary of the stopping region $\stopRegionBound$ is a P-null set.
To this end, the functions used in \cref{theo:solutionOptimalStopping}, the probability measure defined in \cref{eq:updateProbMeasure} and the sets defined in \cref{{eq:defRegionsn},{eq:defRegionsN}} are transferred to another domain. Instead of only depending on the sufficient statistic $\stat{n}$, these quantities are now defined in such a way that they depend on the posterior probabilities and the sufficient statistic. This trick results in more elegant proofs, since the linearity in the posterior probabilities can be exploited directly.

The posterior probability of hypothesis $\Hyp_i$ is denoted by $\postProb{i}$, $i=0,1$. 
Using this representation, the cost for deciding in favor of hypothesis $\Hyp_i$ can be written as
\begin{align}\label{eq:auxVarTile}
 \auxVarCostOptTilde{i,n} = C_{1-i}\postProb{1-i} + C_{2+i}\postProb{i}\Var[\paramRV\given\stat{n},\Hyp_i]\,.
\end{align}
The posterior probabilities are collected in the tuple $\postProbwo_{n}=(\postProb{0},\postProb{1})$, which is defined on $\left(\stateSpacePostProb, \metricPostProb\right)$.
The transition kernel relating $\postProbwo_n$ and $\postProbwo_{n+1}$ is given by $\postProbwo_{n+1}=\transkernelPostVar{\postProbwo_{n}}{x_{n+1}}{\stat{n}}$.

Now, the functions used in \cref{theo:solutionOptimalStopping} can be written as
\begin{align}
 \begin{split}
  \tilde \rho_n(\stat{n},\postProbwo_{n}) &= \min\bigl\{\tilde g(\stat{n},\postProbwo_{n})\,,\,\tilde d_n(\stat{n}, \postProbwo_{n}) \bigr\}, \quad n < N\\
  \tilde \rho_N(\stat{N},\postProbwo_{N} ) & = \tilde g(\stat{N}, \postProbwo_{N})\\
 \end{split}
\end{align}
with
\begin{align}
 \begin{split}
  \tilde g(\stat{n},\postProbwo_{n}) & = \min\bigl\{C_1\postProb{1} + C_2\postProb{0}\Var[\paramRV\given\stat{n},\Hyp_0]\,,\,C_0\postProb{0} + C_3\postProb{1}\Var[\paramRV\given\stat{n},\Hyp_1]\bigr\} \\
  \tilde d(\stat{n},\postProbwo_{n}) & = 1 + \int \tilde\rho_{n+1}(\xi_{\stat{n}}(x_{n+1}),\transkernelPostVar{\postProbwo_{n}}{x_{n+1}}{\stat{n}}) p(x_{n+1} \given \stat{n} ) \dd x_{n+1}\,.
 \end{split}
\end{align}
The pendant of the probability measure defined in \cref{eq:updateProbMeasure} is given by
\begin{align*}
 \updateProbMeasureTilde(B\times \tilde B) := P\left(\left\{\xnew\in \stateSpaceObs: \xi_{\stat{n}}(\xnew)\in B, \transkernelPostVar{\postProbwo_{n}}{\xnew}{\stat{n}} \in \tilde B\right\} \bgiven \stat{n} \right)\,,
\end{align*}
for all elements $B$ of the Borel $\sigma$-algebra on $\stateSpaceStat$ and all elements $\tilde B$ of the Borel $\sigma$-algebra on $\stateSpacePostProb$. The equivalent of the boundary of the stopping region defined in \cref{{eq:defRegionsn}} is given by
\begin{align}
 \begin{split}
   \stopRegionBoundTilde = & \left\{\left(\stat{n}, \postProbwo_{n} \right)\in E_\statWOa \times \stateSpacePostProb : \tilde g(\stat{n}, \postProbwo_{n}) = 1+ \int \tilde \rho_{n+1}\dd\updateProbMeasureTilde \right\}\,. \\
 \end{split}
\end{align}
\begin{lemma}\label{lem:gBounded}
Let $a=(a_0,a_1)$ and let $a\cdot\postProbwo$ denote the element-wise product. Then for all $a\in\nonNegSet^2$, all $\stat{n}\in\stateSpaceStat$ and all $\postProbwo\in\stateSpacePostProb$, it holds that
 \begin{align*}
  \min\{a_0,a_1, 1\} \tilde g(\stat{n},\postProbwo) \leq \tilde g(\stat{n},a\cdot\postProbwo) \leq \max\{a_0,a_1,1\} \tilde g(\stat{n},\postProbwo)\,, \quad 0\leq n \leq N.
 \end{align*}
\end{lemma}
\begin{proof}
 The proof is given only for the lower bound because the upper bound can be proven analogously. For $a_1\geq a_0$ it holds that
 \begin{align*}
  \tilde g(\stat{n},a\cdot\postProbwo) & = \min\bigl\{a_1C_1\postProb{1} + a_0C_2\postProb{0}\Var[\paramRV\given\stat{n},\Hyp_0]\,,\,a_0C_0\postProb{0} + a_1C_3\postProb{1}\Var[\paramRV\given\stat{n},\Hyp_1]\bigr\} \\
  & = a_1 \min\biggl\{C_1\postProb{1} + \frac{a_0}{a_1}C_2\postProb{0}\Var[\paramRV\given\stat{n},\Hyp_0]\,,\,\frac{a_0}{a_1}C_0\postProb{0} + C_3\postProb{1}\Var[\paramRV\given\stat{n},\Hyp_1]\biggr\} \\
  & \geq a_1 \tilde g(\stat{n},\postProbwo) \geq a_0\tilde  g(\stat{n},\postProbwo)\,.
 \end{align*}
For $a_0\geq a_1$ it further holds that
\begin{align*}
 \tilde g(\stat{n},a\cdot\postProbwo) & = \min\bigl\{a_1C_1\postProb{1} + a_0C_2\postProb{0}\Var[\paramRV\given\stat{n},\Hyp_0]\,,\,a_0C_0\postProb{0} + a_1C_3\postProb{1}\Var[\paramRV\given\stat{n},\Hyp_1]\bigr\} \\
  & = a_0 \min\biggl\{\frac{a_1}{a_0}C_1\postProb{1} + C_2\postProb{0}\Var[\paramRV\given\stat{n},\Hyp_0]\,,\,C_0\postProb{0} + \frac{a_1}{a_0}C_3\postProb{1}\Var[\paramRV\given\stat{n},\Hyp_1]\biggr\} \\
  & \geq a_0 \tilde g(\stat{n},\postProbwo) \geq a_1\tilde  g(\stat{n},\postProbwo)\,.
\end{align*}
This yields $\min\{a_0,a_1,1\}\tilde g(\stat{n},\postProbwo) \leq \tilde g(\stat{n},a\cdot\postProbwo)$ which is the lower bound stated in \cref{lem:gBounded}. 
\end{proof}

\begin{lemma}\label{lem:rhoBounded}
Let $a=(a_0,a_1)$ and let $a\cdot\postProbwo$ denote the element-wise product. Then for all $a\in\nonNegSet^2$, all $\stat{n}\in\stateSpaceStat$ and all $\postProbwo\in\stateSpacePostProb$, it holds that
 \begin{align*}
  \min\{a_0,a_1, 1\} \tilde \rho(\stat{n},\postProbwo) \leq \tilde \rho(\stat{n},a\cdot\postProbwo) \leq \max\{a_0,a_1,1\} \tilde\rho(\stat{n},\postProbwo)\, \quad 0\leq n \leq N.
 \end{align*}
\end{lemma}
The proof of \cref{lem:rhoBounded} is shown in \cref{app:proofLemRhoBounded}.

The result stated in \cref{lem:rhoBounded} is now used to show that the boundary of the stopping region $\stopRegionBoundTilde$ is a P-null set. This statement is fixed in the following Lemma.

\begin{lemma}\label{lemma:boundaryNullSet}
 If the random variables $\postProb{i}$, $i\in\{0,1\}$, are continuous random variables, then the boundary of the stopping region $\stopRegionBoundTilde$ is a P-null set, i.e.,
 \begin{align*}
  \updateProbMeasureTilde(\stopRegionBoundTilde) = 0\quad \forall (\stat{n}, \postProbwo_{n}) \in \stateSpaceStat \times \stateSpacePostProb,\quad \forall n<N\,.
 \end{align*}
\end{lemma}
\begin{proof}
\cref{lemma:boundaryNullSet} can be proven by contradiction. Assume that there exist a non-zero probability $\updateProbMeasureTilde(\stopRegionBoundTilde) > 0$ that the test hits the boundary $\stopRegionBoundTilde$ for some $n<N$ with the next update. Since the variables $\postProb{i}$ and $\stat{n}$ are continuous random variables, there has to exist an interval $[\stat{n}^\bullet, b\stat{n}^\bullet]\times[\postProbwo_n^\bullet, a\postProbwo_n^\bullet]$ with the scalars $a>1$, $b>1$, $n<N$ for which the costs for stopping and continuing the test are equal. 
This implies that
\begin{align}\label{eq:assCostEqual}
 \tilde g(\stat{n}, \postProbwo_n) = 1+ \int \tilde \rho_{n+1}\dd\updateProbMeasureTilde \quad \forall (\stat{n},\postProbwo_n) \in [\stat{n}^\bullet, b\stat{n}^\bullet]\times[\postProbwo_n^\bullet, a\postProbwo_n^\bullet]\,.
\end{align}
Assume for now, that $\stat{n}\in[\stat{n}^\bullet, b\stat{n}^\bullet]$ is fixed. Using the previous results, it follows that
\begin{align*}
 1+ \int \tilde \rho_{n+1}\dd\updateProbMeasureTildeScaled{a} & \leq  1 + a\int \tilde \rho_{n+1}\left(\xi_{\stat{n}}(\xnew), \transkernelPostVar{\postProbwo_{n}}{\xnew}{\stat{n}} \right) p(\xnew\given\stat{n}) \dd \xnew \\
 & <   a + a\int \tilde \rho_{n+1}\left(\xi_{\stat{n}}(\xnew), \transkernelPostVar{\postProbwo_{n}}{\xnew}{\stat{n}} \right) p(\xnew\given\stat{n}) \dd \xnew \\
 & =  a\left( 1 + \int \tilde \rho_{n+1} \dd\updateProbMeasureTilde \right)\\
 & =  a \tilde g(\stat{n}, \postProbwo_n) \leq \tilde g(\stat{n}, a \postProbwo_n)\,,
\end{align*}
where the first and the last inequality are due to \cref{lem:rhoBounded,lem:gBounded}, respectively.
The assumption of equal costs for stopping and continuing the test does not hold for a fixed $\stat{n}\in[\stat{n}^\bullet, b\stat{n}^\bullet]$.
This contradicts the assumption made in \cref{eq:assCostEqual}, that the costs for stopping and continuing the test are equal for all $(\stat{n},\postProbwo_n) \in [\stat{n}^\bullet, b\stat{n}^\bullet]\times[\postProbwo_n^\bullet, a\postProbwo_n^\bullet]$. 
Since $\stopRegionBoundTilde$  is only a different representation for $\stopRegionBound$, \cref{lemma:boundaryNullSet} entails also that $\updateProbMeasure(\stopRegionBound)=0$ for all $\stat{n}\in\stateSpaceStat$.
\end{proof}
\cref{lemma:boundaryNullSet} implies that there is no need for a randomized stopping rule since for all $\stat{n}\in\stateSpaceStat$ and all $n<N$ the cost minimizing stopping rule is exactly defined. This result is essential for the theorems shown in the remainder of this work.

To present the upcoming results in a more compact way, the following short-hand notations are used
\begin{align*}
 z_n^i = \frac{p(\stat{n}\given\Hyp_i)}{p(\stat{n})} \quad \text{and} \quad {\{\xi_{\stat{n}}\in\stopRegionCompl[n+1]\}} := \{\xnew \in \stateSpaceObs: \xi_{\stat{n}}(\xnew)\in\stopRegionCompl[n+1]\}\,.
\end{align*}

\begin{theorem}\label{theo:costDerivatives}
Let $\rho^\prime_{n,C_i}$ denote the derivative of $\rho_n$ with respect to $C_i$ for $i\in\{0,1,2,3\}$, and let 
\begin{align} \label{eq:defStopRegionDec}
 \stopRegionDec{n}{i} := \stopRegion \cup \left\{\auxVarCostOpt{i,n} \leq \auxVarCostOpt{{1-i},n}\right\}
\end{align}
be the region in which the test stops at time $n$ and decides in favor of hypothesis $\Hyp_i$.
For $i\in\{0,1\}$ and $n<N$, it holds that
\begin{align*}
  \rho^\prime_{n,C_i}(\stat{n}) =\left\{
		     \begin{array}{ll}
		      p(\Hyp_i) z_n^i&\text{for } \stat{n}\in \stopRegionDec{n}{1-i}\\
		      r^i_n(\stat{n}) &\text{for } \stat{n}\in \stopRegionCompl \\
		      0 & \text{for } \stat{n}\in \stopRegionDec{n}{i}\\
		     \end{array}
		     \right.
 \end{align*}
where $r^i_n$ is defined recursively via
\begin{align*}
 r^i_n(\stat{n}) = p(\Hyp_i) z_n^i \updateProbMeasure^i\bigl(\stopRegionDec{n+1}{1-i}\bigr) + \int_{\stopRegionCompl[n+1]} \rho_{n+1,C_i}^\prime  \dd\updateProbMeasure\,.
\end{align*}
For $i\in\{0,1\}$ and $n=N$, it holds that:
\begin{align*}
  \rho^\prime_{N,C_i}(\stat{n}) =\left\{
		     \begin{array}{ll}
		      p(\Hyp_i) z_N^i&\text{for } \stat{N}\in \stopRegionDec{N}{1-i}\\
		      0 & \text{for } \stat{N}\in \stopRegionDec{N}{i}\\
		     \end{array}
		     \right.  
 \end{align*}
For $i\in\{2,3\}$ and $n<N$, it holds that
\begin{align*}
  \rho^\prime_{n,C_i} =\left\{
		     \begin{array}{ll}
		      p(\Hyp_{i-2})z^{i-2}_{n}\Var\left[\paramRV\given\stat{n},\Hyp_{i-2}\right] &\text{for } \stat{n}\in \stopRegionDec{n}{i-2}\\
		      r^i_n(\stat{n}) &\text{for } \stat{n}\in \stopRegionCompl \\
		      0 & \text{for } \stat{n}\in \stopRegionDec{n}{3-i}\\
		      \end{array}
		     \right.  
 \end{align*}
 with
 \begin{align*}
  r_n^i(\stat{n}) = & p(\Hyp_{i-2})z^{i-2}_{n}  \int_{\{\xi_{\stat{n}} \in \stopRegionDec{n+1}{i-2}\}} \Var\left[\paramRV\given\xi(\stat{n},\xnew),\Hyp_{i-2}\right]  p(\xnew\given\stat{n},\Hyp_{i-2}) \dd\xnew + \int_{\stopRegionCompl[n+1]} \rho_{n+1,C_i}^\prime  \dd\updateProbMeasure\,.
 \end{align*}
For $i\in\{2,3\}$ and $n=N$ it holds that:
\begin{align*}
  \rho^\prime_{n,C_i} =\left\{
		     \begin{array}{ll}
		      p(\Hyp_{i-2})z^{i-2}_{N}\Var\left[\paramRV\given\stat{N},\Hyp_{i-2}\right] &\text{for } \stat{N}\in \stopRegionDec{N}{i-2}\\
		      0 & \text{for } \stat{N}\in \stopRegionDec{N}{3-i}\\
		      \end{array}
		     \right.  
 \end{align*}
\end{theorem}
A proof of \cref{theo:costDerivatives} is given in \cref{app:proofCostDerivatives}.

\cref{theo:costDerivatives} is now used to show the connection between the cost functions $\rho_n$ and the performance measures of the test.
\begin{theorem}\label{theo:derivativesPerformanceRelation}
 For the derivatives $\rho^\prime_{n,C_i}$ defined in \cref{theo:costDerivatives} and for the optimal policy stated in \cref{cor:optPolicyC}, it holds that
 \begin{align*}
  \rho^\prime_{n,C_i}(\stat{n}) & = p(\Hyp_i) z_n^i \errorDetOpt{n}{i}(\stat{n}) \quad\quad\quad\;\;\, i\in\{0,1\}\,, \\
  \rho^\prime_{n,C_i}(\stat{n}) & = p(\Hyp_{i-2}) z_n^{i-2} \errorEstOpt{n}{i-2}(\stat{n}) \quad  i\in\{2,3\}\,,
 \end{align*}
 and in particular
 \begin{align*}
  \rho^\prime_{0,C_i}(\stat{0}) & = p(\Hyp_i) \errorDetOpt{0}{i}(\stat{0}) \quad\quad\;\;\; i\in\{0,1\}\,, \\
  \rho^\prime_{0,C_i}(\stat{0}) & = p(\Hyp_{i-2}) \errorEstOpt{0}{i-2}(\stat{0}) \quad  i\in\{2,3\}\,.
 \end{align*}
\end{theorem}
A proof of \cref{theo:costDerivatives} is laid down in \cref{app:proofDerivativesPerformanceRelation}.
This is the main property of the cost function obtained from the optimal stopping problem and it is fundamental for determining a set of optimal weights to guarantee a predefined performance of the test.

\section{Optimal Coefficients of the Loss Function}\label{sec:optCostCoeff}
The question how to choose the weights of the individual costs in a Bayesian cost function has attained little attention in the literature and the choice is usually left to the practitioner.
In a pure fixed-sample size detection problem, it is rather simple to choose the correct weights because the ratio $C_0/C_1$ is a trade-off between the type I and type II error probabilities. In the case of sequential detection, the choice the coefficients becomes more difficult since the values of $C_0$ and $C_1$ are a trade-off between the run-length of the test and the two error probabilities.

For the sequential joint detection and estimation problem, it becomes almost impossible to choose the parameters by hand, since the corresponding performance measures typically have different scales altogether. The error probabilities are in the interval $[0,1]$, whereas the estimation quality when using a squared error loss is in the range $[0,\infty)$.

One way to overcome this, is to normalize the estimation error to the unit interval, see, e.g., \citet{zhang2009dynamic}.
The normalization constant can, for example, be found by using training data. Nevertheless, this leaves the problem open how to chose the trade-off between expected run-length, error probabilities and estimation performance.

In this section, a strategy how to obtain the coefficients based on linear programming is presented. It adopts the approach given in \citet{fauss2015linear} for the sequential joint detection and estimation problem. Recall the problem formulation:
\emph{design a sequential procedure which uses on average as few samples as possible and fulfills constraints on the error probabilities and the estimation quality.}

In technical terms, this problem can be formulated as the following constrained optimization problem, where it is assumed that $N$ is sufficient large so that the constraints can be fulfilled.

\begin{problem} \label{prb:constrProblem}
    Assuming a sufficient large $N$, the optimal policy for the sequential joint detection and estimation problem $\policyOptErr$ can be found by solving the following constrained optimization problem
    \begin{align}
        \begin{split}
            & \min_{\policy\in\policySet}\; \E[\tau], \quad \stopR_N = 1\\
            \text{s.t.} \qquad & p(\dec_\tau=1\given\Hyp_0) \leq \errConstr_0\,, \\
            & p(\dec_\tau=0\given\Hyp_1) \leq \errConstr_1\,, \\
            & \E\bigl[\ind{\dec_\tau=0}(\est{0,\tau}-\theta)^2\given \Hyp_0 \bigr] \leq \errConstr_2\,, \\
            & \E\bigl[\ind{\dec_\tau=1}(\est{1,\tau}-\theta)^2\given  \Hyp_1 \bigr] \leq\errConstr_3\,,
        \end{split}
    \end{align}
    where $\policySet$ is the set of all feasible policies and $\errConstr_0,\errConstr_1\in(0,1)$ and $\errConstr_2,\errConstr_3\in(0,\infty)$.
\end{problem}
This constrained optimization problem is a rather intuitive formulation.

In order to obtain $\policyOptErr$, the strong connection between the derivatives of the cost function and the performance measures as stated in \cref{theo:derivativesPerformanceRelation} is exploited.

The following maximization problem, which is in fact the dual problem of the one stated in \cref{prb:constrProblem} (see \cref{app:proofTheoMaxC} for details), is considered

\begin{align}\label{eq:optMaxC}
 \max_{C\geq0}\, L_\errConstr(C)\,,
\end{align}
with $C=\{C_0, C_1, C_2, C_3\}$ and the dual objective given by
\begin{align}
 L_\errConstr(C) = \rho_0(\stat{0}) - \sum_{i=0}^1 p(\Hyp_i)C_i\errConstr_i - \sum_{i=2}^3 p(\Hyp_{i-2})C_i\errConstr_i\,.
\end{align}

The coefficients $C_i$, $i\in\{0,1,2,3\}$, now act as Lagrange multipliers for the inequalities in \cref{prb:constrProblem} and $L_\errConstr(C)$ is the Lagrangian dual objective.
Since it is not trivial that \cref{eq:optMaxC} is indeed the dual problem of \cref{prb:constrProblem}, we state the following theorem.

\begin{theorem}\label{theo:maxC}
Let $\policyOptErr$ be the solution of \cref{prb:constrProblem} , let $C_{\errConstr}^\star$ be the solution of \cref{eq:optMaxC} and let $\policyOptCerr$ be the optimal policy parametrized by $C_\errConstr^\star$. Then, it holds that
\begin{align*}
 \policyOptCerr = \policyOptErr \quad \text{and} \quad L_\errConstr(C_{\errConstr}^\star) = \E[\tau(\policyOptCerr)]\;.
\end{align*}
That is, if a policy solves \cref{eq:optMaxC} then, it also solves \cref{prb:constrProblem}. Moreover, the optimal value of \cref{eq:optMaxC} is the expected run-length.
\end{theorem}
A proof of \cref{eq:optMaxC} can be found in \cref{app:proofTheoMaxC}.
Using \cref{theo:maxC} and \cref{eq:optMaxC}, \cref{prb:constrProblem} can be shown to be equivalent to the following maximization problem:
\begin{align}\label{eq:finalOptimizationProblem}
   & \max_{C\geq0}\,\rho_0(\stat{0}) - \sum_{i=0}^1 p(\Hyp_i)C_i\errConstr_i - \sum_{i=2}^3 p(\Hyp_{i-2})C_i\errConstr_i\\
  \text{s.t.} \quad  & \rho_n(\stat{n}) =  \min\left\{ \costS(\stat{n}),1 +  \int \rho_{n+1}\dd\updateProbMeasure  \right\} \quad n<N\nonumber\\
    & \rho_N(\stat{N}) = \costS(\stat{N})\nonumber
\end{align}
In order to solve the optimization problem in \cref{eq:finalOptimizationProblem}, we proceed, as in \citet{fauss2015linear}, by relaxing the equality constraints to inequality constraints and adding the cost function $\rho_n$ to the set of free variables.

\begin{theorem}\label{theo:jointDetEstLP}
The problem
 \begin{align} \label{eq:jointDetEstLP}
  \begin{split}
   \max_{C\geq0,\rho_n\in\mathcal{L}}\;&\rho_0(\stat{0}) - \sum_{i=0}^1 p(\Hyp_i)C_i\errConstr_i - \sum_{i=2}^3 p(\Hyp_{i-2})C_i\errConstr_i  \\
  \text{s.t.} \quad  & \rho_n(\stat{n}) \leq  \min\left\{ \auxVarCostOpt{0,n}(\stat{n}),\; \auxVarCostOpt{1,n}(\stat{n}),\; 1+\int \rho_{n+1}\dd\updateProbMeasure \right\}\quad\text{for } 0\leq n < N \\
  & \rho_N(\stat{N}) \leq  \min\left\{ \auxVarCostOpt{0,N}(\stat{N}),\; \auxVarCostOpt{1,N}(\stat{N})\right\}
  \end{split}
 \end{align} is equivalent to \cref{prb:constrProblem}, where $\mathcal{L}$ is the set of all non-negative $\updateProbMeasure$-integrable functions on $E_\statWOa$ and $\auxVarCostOpt{i,n}$ is defined in \cref{eq:defAuxCostOpt}. 
\end{theorem}
\begin{proof}
Let $\rho^\star$ denote the solution of problem \eqref{eq:finalOptimizationProblem} and let $\rho^\dag$ denote the solution of the corresponding relaxed problem \eqref{eq:jointDetEstLP}. Since problem \eqref{eq:jointDetEstLP} is a relaxed version of problem \eqref{eq:finalOptimizationProblem}, it holds that
\begin{align}\label{eq:rho0IneqRel}
 \rho^\star_0(\stat{0}) \leq \rho^\dag_0(\stat{0})\,.
\end{align}
The function relating $\rho_n$ and $\rho_{n+1}$, i.e., $\rho_n = F(\rho_{n+1})$, is monotonically non-decreasing. This function consists of a minimum operator, which is monotonic by definition. The argument of the minimum operator is an expected value of a non-negative function, which implies that the function $F(\cdot)$ can never decrease when its argument increases. Assume that, without changing $C_\kappa^\star$, the inequality constraint in \cref{eq:jointDetEstLP} is not fulfilled with equality for some $n$, i.e., 
\begin{align*}
 \rho^\dag_n(\stat{n}) \leq \rho^\star_n(\stat{n})\,.
\end{align*}
Due to the monotonicity of $F(\cdot)$, we can state that
\begin{align*}
 \rho^\dag_n(\stat{n}) \leq \rho^\star_n(\stat{n}) \Rightarrow \rho^\dag_{n-1}(\stat{n-1}) \leq \rho^\star_{n-1}(\stat{n-1})\,.
\end{align*}
Applying this relations recursively gives us
\begin{align} \label{eq:rho0IneqMon}
 \rho^\dag_0(\stat{n}) \leq \rho^\star_0(\stat{n})\,.
\end{align}
Combining \cref{eq:rho0IneqMon} and \cref{eq:rho0IneqRel} yields $\rho^\dag_0(\stat{0}) = \rho^\star_0(\stat{0})$. This implies that $\rho^\dag$ and $\rho^\star$ can differ only in a P-null set. Hence, the policies corresponding to $\rho^\dag$ and $\rho^\star$ are equivalent in an almost sure sense.
\end{proof}
 \section{Discussion}\label{sec:discussion}
We derived an optimal sequential joint detection and estimation scheme, starting from a decision theoretical point of view and ending with a recursively defined cost function. 
To ensure a certain performance in terms of error probabilities and estimation errors, a method how to choose the weights in the joint loss function is presented. By adding the cost function $\rho_n$ itself to the set of free variables, one ends up with a problem, which is linear in the coefficients $C_i$ as well as in the cost function $\rho_n$. Though this problem is linear, an optimization over a functional, i.e., an infinite dimensional, space has to be performed.

Optimization over infinite dimensional spaces is in general a challenging task, which is, for example, addressed in \citet{botelho2016functional}. Although there exist efficient solution approaches for this kind of problems, they are beyond the scope of this paper. For the examples presented in \cref{sec:numResults}, a straightforward discretization using regular grids turns out to be sufficient.
After a discretization of the state spaces of the parameters, the observations and the sufficient statistic, this problem reduces to a finite dimensional linear program which can be solved by a variety of powerful off-the-shelf solvers.
From the solution of the optimization problem, one directly obtains discretized versions of the cost functions $\rho_n$ and $d_n$ as well as the set of optimal coefficients $C_{\errConstr}^\star$, which are needed for the decision rule.

The optimal objective $L_\errConstr(C_{\errConstr}^\star)$ corresponds to the expected number of samples used by the sequential scheme. 
Most design procedures for sequential schemes do not provide this information during the design process and have to be estimated by e.g. Monte Carlo simulations after the scheme is designed.

Issues about numerical stability can arise when one tries to solve \cref{eq:jointDetEstLP} using a straightforward discretization because the contribution of $\rho_0(\stat{0})$ to $\rho_n(\stat{n})$ becomes smaller with increasing $n$. Especially when this contribution is close to the solvers accuracy, an accurate solution of $\rho_n$ can no longer be guaranteed.
In general, this contribution is low in regions that are very unlikely and thus should not effect the stopping regions. Depending on the likelihood, the unknown parameters and the corresponding prior, this could also happen in regions that will affect the stopping regions. Therefore, we suggest to use a regularization term to ensure that the cost function $\rho_n$ is indeed maximized over the entire state space and, hence, the inequality constraints are fulfilled with equality. 
The regularized formulation does not result in a strictly optimal test,
anymore. Nevertheless, the resulting test differs only slightly from the strictly optimal one. A detailed formulation and further analysis of the regularized formulation can be found in \cref{app:regularizedFormulation}

Once the test is designed, one only has to update the sufficient statistic $\stat{n}$ with every new sample and to evaluate the cost functions $\rho_n$ and $d_n$ at the current value of $\stat{n}$. This can be done, for example, by mapping $\stat{n}$ to the discretization grid to obtain the values of the cost functions or to interpolate the cost functions at point $\stat{n}$. Regardless of which of the two methods is used, the implementation is more cost-efficient than the comparable approaches, such as the one by \citet{Yilmaz2016Sequential}. Whereas our approach only requires the evaluation of the two cost functions at every time instant to decide whether to stop or to continue sampling, the approach by Yilmaz et al. requires at every time instant the computation of the estimates under each hypothesis as well as the evaluation of the decision rule to get the overall costs. This is computationally  much more expensive than the evaluation of two costs functions. Hence, the easy and fast implementation of the test, to an extent, compensates the costly test design.

Another advantage of the proposed approach compared to heuristic approaches for choosing the weights $C_i$ is that a detection (estimation) coefficient is set to zero if the constraint is implicitly fulfilled by the corresponding estimation (detection) constraint of the same hypothesis, see \cref{sec:numResults} and \cref{app:proofTheoMaxC} for details. The optimal procedure then becomes equivalent to a test which is designed under estimation (detection) constraints only. In contrast, it is rather impossible to assess the implicit fulfillment of the constraint heuristically beforehand.

Although this paper focuses on the \ac{MSE} as a quality measure for estimation, any appropriate loss function can be used as long as the same quality measure is used in the constrained problem formulation (\cref{prb:constrProblem}). Using a different quality measure leads to a different policy, since the \ac{MMSE} estimator is not optimal anymore, but the general method proposed in this work can still be used and all proofs for optimality remain valid. Using, for example, the absolute error as a quality measure for the estimator would lead to an optimal sequential scheme with the posterior median as an optimal estimator.
Also the restriction that the estimation quality is the same under both hypotheses is only made for the sake of a compact notation and is not required in general.

Last but not least, it should be mentioned that the proposed method can in general be extended to multiple hypotheses. The loss function stated in \cref{eq:lossFctJoint} can either be defined in terms of pairwise error probabilities, i.e., $P(\dec_\tau = i \given \Hyp_j)$, in terms of the family-wise errors or in terms of the false discovery rate \citep{benjamini1995controlling}. The choice how to penalize wrong decisions influences the cost for deciding in favor of $\Hyp_i$, i.e., $\auxVarCostOpt{i}$. Obviously, the decision rule given in \cref{eq:optDecRule} becomes an $\argmin$ of the costs for deciding in favor of the different hypotheses. Using this decision rule, the rest of the analysis presented in this paper essentially stays the same.
 \section{Numerical Results}\label{sec:numResults}
As a proof-of-concept and to show the properties of the resulting tests, two numerical examples are presented in this section. To solve the problems numerically, the state spaces of the observations, the parameter  and  the sufficient statistic---which are continuous spaces---are discretized using regular grids. On these grids, the optimization problem in \cref{eq:jointDetEstLPReg} is solved to obtain the optimal weights and the cost functions.
Let $\Nx$ and $\Nstat$ denote the number of grid points of the observation space and the space of the sufficient statistic, respectively. Furthermore, let $\boldsymbol\rho_n$ and $\auxVarCostOptDiscr{i,n}$ be the discretized versions of $\rho_n(\stat{n})$ and $\auxVarCostOpt{i,n}$, respectively. Both are column vectors of size $\Nstat\times 1$.
For the look ahead step, the matrix $\boldsymbol\xi_n$ of size $\Nx\times\Nstat$ is introduced which represents the transition kernel $\xi_{\stat{n}}$. 
The discrete version of the posterior predictive $p(\xnew\given\stat{n})$ is denoted by $\boldsymbol P_n$ and is of size $\Nstat\times\Nx$. Hence, the probability measure $\updateProbMeasure$ is represented by $\updateProbMeasureMtx = \Delta x\boldsymbol P_n \boldsymbol \xi_n$, where $\Delta x$ is the distance of two elements of the grid of $x$. Let  $\auxVarCostOptMtx{i} = [\auxVarCostOptDiscr{i,0}, \ldots, \auxVarCostOptDiscr{i,N}]$, $\boldsymbol\rho=[\boldsymbol\rho_0,\ldots,\boldsymbol\rho_N]$ and let $\rho_n^m$ denote the $m$-th element of $\boldsymbol\rho_n$.
The discretized version of problem~\eqref{eq:jointDetEstLP} can then be written as the linear program:
\begin{align}\begin{split}\label{eq:jointDetEstLPDiscrete}
   \max_{C\in\nonNegSet^{4},\boldsymbol\rho\in\nonNegSet^{\Nstat\times N+1}}\;&\rho_0^{m} - \sum_{i=0}^1 p(\Hyp_i)C_i\errConstr_i - \sum_{i=2}^3 p(\Hyp_{i-2})C_i\errConstr_i \\
  \text{s.t.} \quad  & \boldsymbol\rho_{\phantom{n}} \leq  \auxVarCostOptMtx{0}\\
  & \boldsymbol\rho_{\phantom{n}} \leq  \auxVarCostOptMtx{1}\\
  & \boldsymbol\rho_n \leq 1+\updateProbMeasureMtx\boldsymbol\rho_{n+1} \quad \text{for } 0\leq n < N 
  \end{split}
\end{align}
where the index $m$ has to be chosen such that $\rho_0^m$ corresponds to $\rho_0(\stat{0})$.

This discretized optimization problem consists of $\Nstat(N+1) +4$ unknown variables and $\Nstat(3N+2)$ inequality constraints.

The maximization problem in \cref{eq:jointDetEstLPDiscrete} is solved using the Gurobi \citep{gurobi} solver which is called via the MATLAB cvx interface \citep{gb08,cvx}.
To reduce the influence of numerical errors, \ac{LP} is only used to obtain the set of optimal coefficients $C_\errConstr^\star$, and the cost functions are then re-calculated using their recursive definition in \cref{theo:solutionOptimalStopping}---see \cref{sec:discussion} for a discussion on the numerical stability.
To evaluate the costs for stopping and continuing when running a test, a linear interpolation of the cost functions is performed.

As a benchmark, the proposed method is compared with a truncated version of the \ac{SPRT} followed by an \ac{MMSE} estimator.
Although the \ac{SPRT} does not take an estimation step into account, it is an asymptotically optimal method for the detection problem. Hence, it is used to show the gap in the performance between a two-step procedure (\ac{SPRT} followed by \ac{MMSE} estimator) and the joint optimal solution. Running the \ac{SPRT} and two estimators---one under each hypothesis---in parallel would be another benchmark possibility. Since this method would result in a complicated fusion of the three different stopping rules and is still not an overall optimal solution, we focus on the two-step procedure as a benchmark.

The stopping rule and the decision rule of the truncated \ac{SPRT} are given by
\begin{align*}
  \stopR_n^\text{SPRT}(\stat{n}) = & \left\{
 \begin{array}{ll}
    1 & \text{for } \likelihoodRatio(\stat{n}) \notin (A,B) \\
    0 & \text{for } \likelihoodRatio(\stat{n}) \in (A,B)
 \end{array}
 \right. \\
 \stopR_N^\text{SPRT}(\stat{n}) = & 1 \\
  \dec_n^\text{SPRT}(\stat{n}) = & \left\{
 \begin{array}{ll}
    1 & \text{for } \likelihoodRatio(\stat{n}) \geq A \\
    0 & \text{for } \likelihoodRatio(\stat{n}) \leq B
 \end{array}
 \right. \\
  \dec_N^\text{SPRT}(\stat{N}) = & \left\{
 \begin{array}{ll}
    1 & \text{for } \likelihoodRatio(\stat{N}) > 1 \\
    0 & \text{for } \likelihoodRatio(\stat{N}) \leq 1
 \end{array}
 \right.
\end{align*}
where $\likelihoodRatio(\stat{n})$ denotes the likelihood ratio, i.e.,
\begin{align*}
 \likelihoodRatio(\stat{n})=\frac{p(\stat{n}\given\Hyp_1)}{p(\stat{n}\given\Hyp_0)}=\frac{\int p(\stat{n}\given\param,\Hyp_1)p(\param\given\Hyp_1)\dd\param}{\int p(\stat{n}\given\param,\Hyp_0)p(\param\given\Hyp_0)\dd\param}\;.
\end{align*}
The stopping time of the truncated \ac{SPRT} is given by
\begin{align*}
 \tau^\text{SPRT} = & \min\{n:\stopR_n^\text{SPRT} = 1\}\,.
\end{align*}
The thresholds $A$ and $B$ are calculated according to \citet{wald1947sequential} as
\begin{align}\label{eq:waldTh}
 A=\frac{1-\errConstr_{1}}{\errConstr_{0}} \quad \text{and} \quad B=\frac{\errConstr_{1}}{1-\errConstr_{0}}\,.
\end{align}
Note that the \ac{SPRT} uses the likelihood ratio instead of the posterior probabilities and hence, \cref{eq:waldTh} uses $\errConstr_i$ instead of $\errConstr_i p(\Hyp_i)$ for calculating the thresholds.
Normally, one would run the \ac{SPRT} as long as none of the thresholds defined in \cref{eq:waldTh} is crossed, but since the optimal scheme presented in this work is truncated, we decided to also truncate the \ac{SPRT} at the same $N$ for the sake of fairness. To make a decision at time $N$, the likelihood ratio has to be compared against a single threshold. 
Increasing(decreasing) the threshold leads to a systematic preference of hypothesis $\Hyp_0$($\Hyp_1$) once the truncation point $N$ is reached.
Here, this threshold is set to $1$, which means that we decide in favor of the hypothesis with the higher likelihood. Although one could influence the empirical error probabilities by changing this threshold, the \ac{SPRT} with the thresholds defined in \cref{eq:waldTh} has typically significant smaller empirical error probabilities than the nominal ones.

The performance of the two setups presented in the following is validated using a Monte Carlo simulation with $10^7$ runs.

A MATLAB implementation of the algorithms that generate the following examples is available at \url{https://github.com/ReinhardDominik/bayesSeqJointDetEst}.
\subsection{Shift-in-Mean Test}\label{sec:exRandomMean}
The first example used to demonstrate the performance of the method is a shift-in-mean test. Therefore, the conditional distribution of the data, given the mean, is a Gaussian distribution with equal and known variances for both hypotheses. The two hypotheses differ only in the mean or, more precisely, in the priors of the mean. Under the null hypothesis, the negated mean is assumed to follow a Gamma distribution, whereas the mean under the alternative follows a Gamma distribution.
Formally, the two hypotheses can be formulated as follows
\begin{align}
 \begin{split}
  \Hyp_0: & \quad \SeqDataRV\given\mu_0 \sim \norm{\mu_0}{\sigma^2},\;-\mu_0\sim \Gam(a,b)\;, \\
  \Hyp_1: & \quad \SeqDataRV\given\mu_1 \sim \norm{\mu_1}{\sigma^2},\;\phantom{-}\mu_1\sim \Gam(a,b)\;,
 \end{split}
\end{align}
where $\Gam(a,b)$ is the Gamma distribution with shape parameter $a$ and scale parameter $b$. 

Since the distribution of the data conditioned on the mean $\mu_i$, $i\in\{0,1\}$, is Gaussian, it can be written as
\begin{align*}
 p(\obs\given\mu_i,\Hyp_i) & = \prod_{l=1}^n p(x_{l}\given\mu_i,\Hyp_i)\,, \\
			   & \propto \exp\left(-\frac{n}{2\sigma^2} \left(\bar{x}_n-\mu_i\right)^2\right)\,, \\
			   & \propto \norm{\bar{x}_n\bbgiven\mu_i}{\frac{\sigma^2}{n}}\,,
\end{align*}
where $\bar{x}_n=\frac{1}{n}\sum_{i=1}^n x_i$. As the variance $\sigma^2$ is fixed, the conditional distribution of the data is fully described by the number of samples $n$ and the sample mean $\bar{x}_n$. Hence, these two quantities are used as sufficient statistic. The update of the sample mean is given by
\begin{align*}
 \bar{x}_{n+1}=\xi(\bar{x}_n,x_{n+1}) = \frac{1}{n+1}( n\bar{x}_n + x_{n+1})\,.
\end{align*}
The variance of the likelihood is chosen to be $\sigma^2=4$, the a priori probabilities of the two hypotheses are $p(\Hyp_0)=p(\Hyp_1)=0.5$ and the parameters of the Gamma distributions are $a=1.7$ and $b=1$.

As mentioned in \cref{lemma:boundaryNullSet}, the posterior probabilities have to be continuous random variables for the boundary of the stopping region to be a P-null set. In this example, the likelihood is continuous in the sufficient statistic $\bar{x}_n$ and in the random mean $\mu$, which follows itself a continuous distribution. Hence, the posterior distribution of the mean is also continuous in the sufficient statistic $\bar{x}_n$ so that the posterior probabilities are continuous random variables. Therefore, the boundary of the stopping region is a P-null set according to \cref{lemma:boundaryNullSet}.

The aim is to design an optimal sequential test which uses at most $100$ samples and with type I and type II error probabilities not exceeding $0.05$ and $0.025$, respectively. Moreover, the \ac{MSE} under $\Hyp_0$ and $\Hyp_1$ should be upper bounded by $0.35$ and $0.2$, respectively.

In order to solve the optimization problem numerically, the state space of the observations $\stateSpaceObs$ is discretized on $[-15,15]$ with $6000$ grid points. The state space of the parameter $\stateSpaceParam$ is discretized on $[-12,12]$ using $4800$ grid points and the state space of the sufficient statistic $\stateSpaceStat$ is discretized on $[-8,8]$ using $1600$ points. All grids are regular. The resulting update matrices $\boldsymbol\xi_n$ are then of dimension $1600\times6000$. Since, depending on the actual statistic $\bar{x}$ and the time instant $n$, the posterior predictive $p(\xnew\given\stat{n})$ can become very sharp and many grid points are required to represent it properly.

Although the choice of the initial statistic $\bar{x}_0$ is arbitrary, because $\bar{x}_1=x_1$ regardless of the choice of $\bar{x}_0$, the initial statistic is set to $\bar{x}_0=0$ for the optimization problem. The regularization constant is set to $\regConst=5\cdot10^{-4}$, which is $\frac{1}{50}$ of the smallest constraint.

The optimal weights obtained via \cref{eq:jointDetEstLP} are $C^\star_\errConstr\approx[125.1, 235.3, 14.9, 74.3]$ and the expected run-length of the test is given by $13.80$ samples.

In \cref{img:RegionsRandomMeanOrig}, the different regions of the optimal test as well as the regions of the corresponding \ac{SPRT} are shown. For $2\leq \bar{x}_n\leq8$ and $n\leq16$, the complement of the stopping region $\stopRegionComplWoIdx$ is dominated by the constraint on the estimation accuracy under $\Hyp_1$. A high arithmetic mean results in a high certainty that $\Hyp_1$ is true but, leads to a very broad posterior distribution for $\mu$  for small sample sizes, resulting in an inaccurate estimate. Since the constraints on the estimation error under $\Hyp_0$ is not as strict as under $\Hyp_1$, this effect is not as pronounced under $\Hyp_0$ as under $\Hyp_1$, but still visible. The dashed line, which represents the thresholds of the \ac{SPRT} does not show this effect at all. A very small or very large arithmetic mean immediately stops the \ac{SPRT}, even for small sample sizes.
For moderate to large sample sizes ($n\geq20$) the stopping region of the optimal test is mainly influenced by the uncertainty about the true hypothesis. 
However, even on regions where the detection constraints dominate, the corridor in which the test should continue is much broader for the \ac{SPRT} than for the jointly optimal  scheme.

\begin{figure}[!t]
    \centering
    \includegraphics{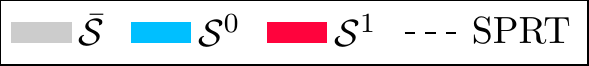}\\
    \subfigure[Original problem]{\includegraphics{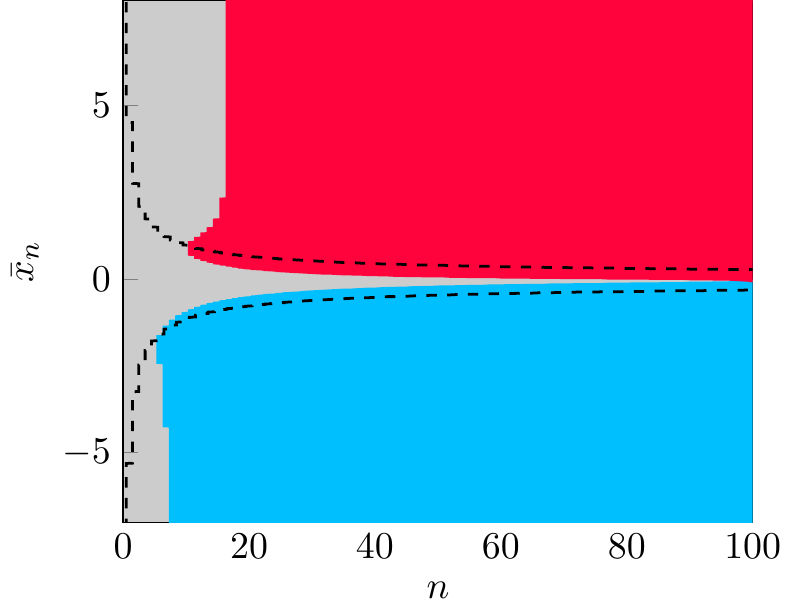}\label{img:RegionsRandomMeanOrig}}
    \hfil
    \subfigure[Relaxed estimation constraints under $\Hyp_1$]{\includegraphics{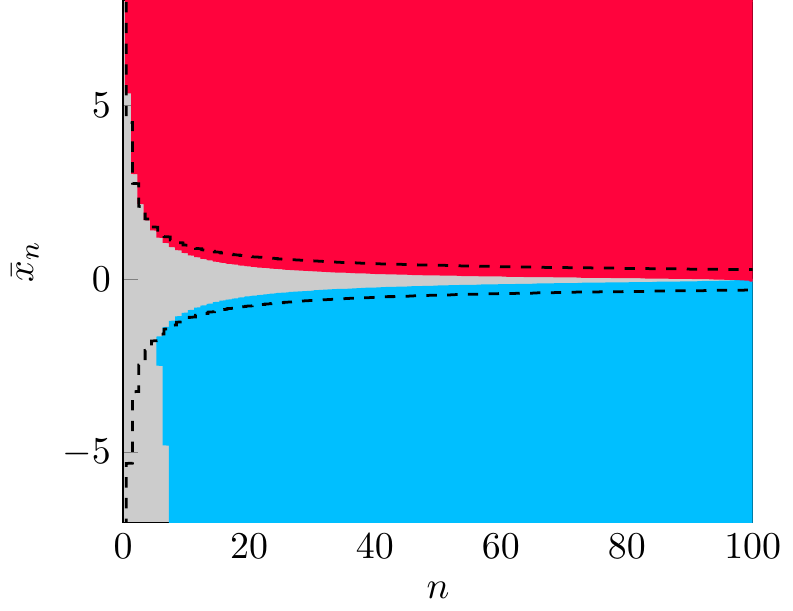}\label{img:RegionsRandomMeanRelaxed}}
    \caption{Different regions of the optimal shift-in-mean test. The region $\stopRegionComplWoIdx$ denotes the complement of the stopping region, the region $\stopRegionDecWoIdx{1}$ is the region where the test stops and decides in favor of $\Hyp_1$, the region $\stopRegionDecWoIdx{0}$ is the region where the test stops and decides in favor of $\Hyp_0$. The dashed lines are the thresholds of the \ac{SPRT}.}
    \label{img:RegionsRandomMean}
\end{figure}

\begin{table}[!t]
 \centering
  \caption{Shift-in-mean test: Constraints and simulation results.}
 \subtable[Detection and estimation errors]{
  \label{tbl:ResultsRandomMeanErr}
  \includegraphics{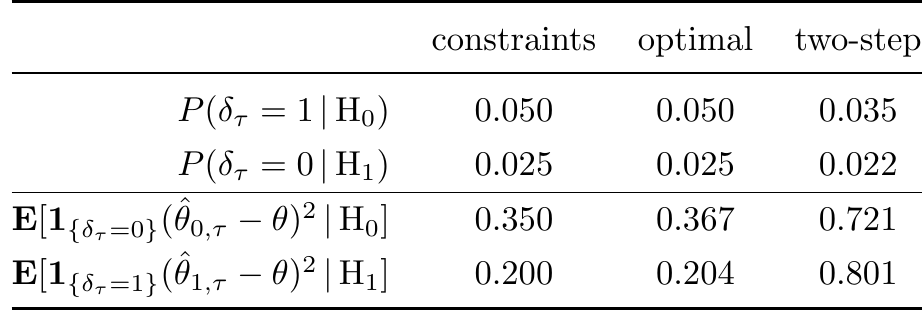}
 }\\
 \subtable[Run-lengths: The second column contains the expected run-length of the optimal scheme obtained as the output of the \ac{LP}.]{
 \label{tbl:ResultsRandomMeanRL}
  \includegraphics{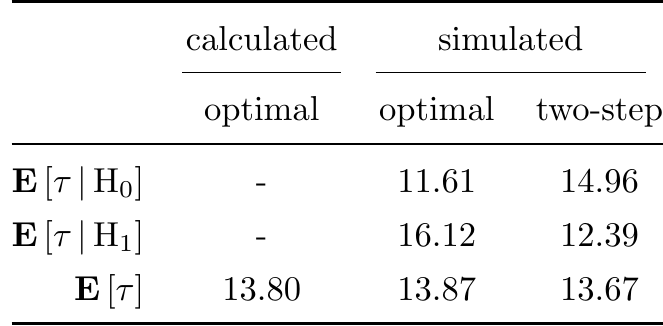}
 }
\end{table}

In \cref{tbl:ResultsRandomMeanErr}, the constraints and the simulation results for both methods, the optimal one and the two-step procedure, are summarized. It can be seen that the optimal sequential joint detection and estimation scheme hits the target error probabilities exactly, whereas the \ac{SPRT} has lower empirical error probabilities. On the other hand, the empirical estimation errors of the optimal scheme are very close to the target ones while the two-stage methods far exceeds the requirements. The difference of the empirical estimation errors of the optimal scheme and the constraints can be explained by numerical inaccuracies and the uncertainty of the Monte Carlo simulation.
Moreover, it can be seen from \cref{tbl:ResultsRandomMeanRL}  that the average run-length of the optimal scheme is very close to the expected run-length that is obtained as an output of the test design. The average run-length of the \ac{SPRT} is slightly smaller, but at the cost of much higher estimation errors. It should be emphasized that the \ac{SPRT} as well as the optimal scheme are truncated at $N=100$ samples. The \ac{SPRT} reaches this point in around $3.4\,\%$ of the cases, whereas this point is only reached in $0.004\,\%$ of the cases for the optimal scheme. Using a regular, non truncated \ac{SPRT} would hence increase the expected run-length and further increase the performance gap of the optimal and the two-step scheme.

In order the show the effect when a detection constraint dominates the corresponding estimation constraint, the estimation constraint under $\Hyp_1$ in the previous example is relaxed to $\errConstr_3=0.9$ while the rest of the setup stays unchanged. 
The optimal weights obtained via \ac{LP} are $C_\errConstr^\star \approx [175.5, 257.8, 14.6, 0]$. As one can see, the optimal weight corresponding to the estimation error under $\Hyp_1$ is set to zero.
This means that the problem is equivalent to a pure detection problem under $\Hyp_1$.
Although one might think that the remaining coefficients are equivalent to the ones of the original problem, this is not the case. All three remaining coefficients have changed in a non-negligible manner.
Hence, even if the optimal coefficients for one problem are known, one cannot use them to approximate the coefficients for the seemingly similar problem. The different regions of the relaxed problem are depicted in \cref{img:RegionsRandomMeanRelaxed}.
Since the estimation constraint under $\Hyp_1$ does not bind anymore, the region $\stopRegionDecWoIdx{1}$ approaches the one of the \ac{SPRT}. By comparing \cref{img:RegionsRandomMeanOrig} and \cref{img:RegionsRandomMeanRelaxed}, it can be seen which subsets of $\stopRegionDecWoIdx{1}$ of the original problem are dominated by the estimation error and ones are dominated by the detection error.

The relaxed problem is also validated by Monte Carlo simulations. Again the \ac{SPRT} followed by a \ac{MMSE} estimator is used for benchmark purposes. The simulation results are summarized in \cref{tbl:ResultsRandomMeanRelaxEstH1Err,tbl:ResultsRandomMeanRelaxEstH1RL}. It can be seen that for the optimal test the nominal error probabilities are hit exactly and the estimation constraint under $\Hyp_0$ is close to the nominal error. The deviation is again caused by numerical inaccuracies. The two-step procedure fulfills the constraints on the detection errors, but exceeds the constraint on the estimation error under $\Hyp_0$.
The relaxed estimation constraint under $\Hyp_1$ is fulfilled for the optimal sequential joint detection and estimation scheme and for the \ac{SPRT}. The optimal scheme has an even lower estimation error under $\Hyp_1$ than the required one while at the same time it uses fewer samples under $\Hyp_1$ than the \ac{SPRT}.

\begin{table}[!t]
 \centering
\caption{Shift-in-mean test with relaxed estimation constraint under $\Hyp_1$: Constraints and simulation results.}
 \subtable[Detection and estimation errors]{
 \label{tbl:ResultsRandomMeanRelaxEstH1Err}
 \includegraphics{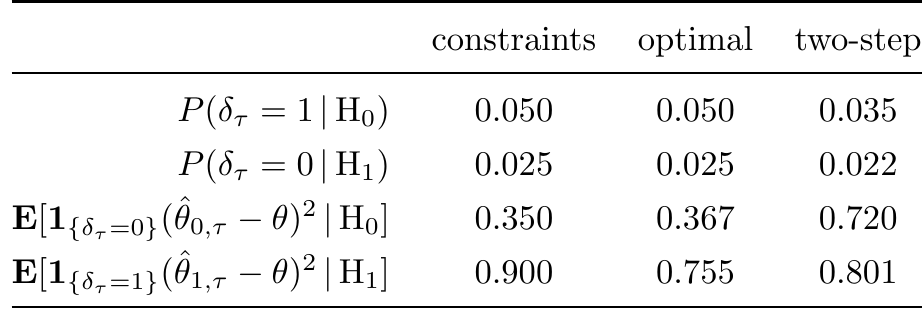}
 }\\
  \subtable[Run-lengths: The second column contains the expected run-length of the optimal scheme obtained as the output of the \ac{LP}.]{
  \label{tbl:ResultsRandomMeanRelaxEstH1RL}
\includegraphics{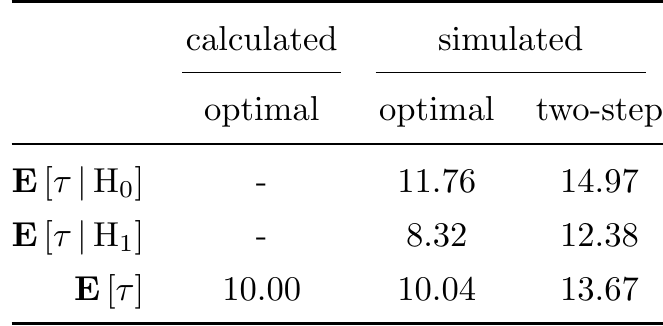}
  }
\end{table}

\subsection{Shift-in-Variance Test}\label{sec:exRandomVariance}
As a second example, a Gaussian likelihood with a random variance is used. In this example, the mean is assumed to be zero under both hypotheses. 
This shift-in-variance formulation is often used when detecting a zero-mean signal in noise.

The distribution of the data conditioned on the variance is Gaussian under both hypotheses and can hence be written as
\begin{align*}
 p(\obs\given\sigma_i^2,\Hyp_i) = & \prod_{l=1}^n p(x_{l}\given\sigma^2_i,\Hyp_i) =  (2\pi\sigma_i^2)^{-\frac{n}{2}} \exp\left( \frac{-ns_n^2}{2\sigma_i^2}\right)\,,
\end{align*}
where
\begin{align*}
 s_n^2 = \frac{1}{n}\sum_{l=1}^n x_l^2\,.
\end{align*}
It can bee seen that the likelihood is fully described by $n$ and $s_n^2$. Therefore, these quantities are used as a sufficient statistic. The corresponding transition kernel is given by $s^2_{n+1}=\frac{1}{n+1}(ns_{n}^2 + x_{n+1}^2)$.

Since, for this example, the likelihood is a continuous function of the random parameter $\sigma_i^2$ and of the sufficient statistic $s_n^2$, one can, using the same line of arguments as in the previous example, show that the boundary of the stopping region is a P-null set according to \cref{lemma:boundaryNullSet}.

Under $\Hyp_0$, i.e., no signal is present, the variance of the observations is assumed to be uniformly distributed on $[\sigma_{0,\text{min}}^2, \sigma_{0,\text{max}}^2]$. Under the alternative, i.e., a signal is present, the variance follows a shifted Gamma distribution with shape and scale parameter $a$ and $b$, respectively. Mathematically, the two hypotheses can be written as
\begin{align}
 \begin{split}
  \Hyp_0: & \quad \SeqDataRV\given\sigma_0^2 \sim \norm{0}{\sigma_0^2}, \quad \sigma_0^2\sim \unif(\sigma_{0,\text{min}}^2,\sigma_{0,\text{max}}^2) \\
  \Hyp_1: & \quad \SeqDataRV\given\sigma_1^2 \sim \norm{0}{\sigma_1^2}, \quad \sigma_1^2-\sigma_{1,\text{min}}^2\sim \Gam(a,b)
 \end{split}
\end{align}
with $\sigma_{0,\text{min}}^2< \sigma_{0,\text{max}}^2 < \sigma_{1,\text{min}}^2$.
In this scenario, the interval for the uniform distribution under $\Hyp_0$ is set to $[0.1, 1]$. Moreover, the parameter $\sigma_{1,\text{min}}^2$ is set to $1.3$ and the shape and scale parameters are set to $1.7$ and $0.5$, respectively.
To ensure that the hypotheses are separable, even after discretization, $\sigma_{1,\text{min}}^2$ is slightly larger than $\sigma_{0,\text{max}}^2$. Both hypotheses have equal prior probability, i.e., $p(\Hyp_0)=p(\Hyp_1)=0.5$.

In order to solve the problem numerically, the state space of the parameter $\stateSpaceParam$ is discretized on $[0.01, 60]$ using $9001$ grid points. The state space of the observations $\stateSpaceObs$ is discretized on $[-20,20]$ with $6000$ grid points and the state space of the sufficient statistic $\stateSpaceStat$ is discretized on $[0, 25]$ using $2100$ grid points. The resulting update matrices $\boldsymbol\xi_n$ are then of dimension $2100\times6000$, which is even larger than in the previous example. This shows the limitations of the straightforward discretization of the continuous spaces.
The resulting sequential test should use at most $N=100$ samples while having target error probabilities of $0.05$ and an \ac{MSE} of $0.025$ and $0.25$ under $\Hyp_0$ and $\Hyp_1$, respectively. To avoid numerical issues with the \ac{LP}, the regularized problem is solved with a regularization constant of $\regConst=5\cdot10^{-5}$.

\begin{figure*}[!t]
    \centering
    \includegraphics{regionLegend}\\
    \subfigure[Original problem, $\errConstr_3=0.25$]{\includegraphics{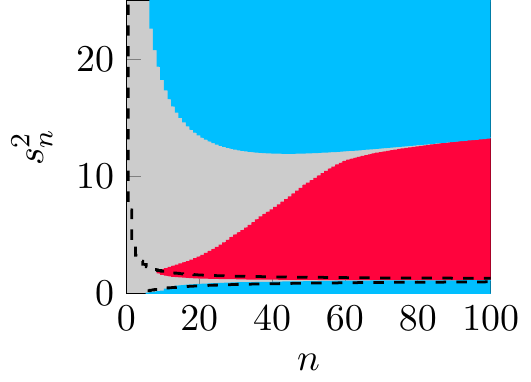}\label{img:RegionsRandomVarOrig}}
     \hfil
    \subfigure[Relaxed constraint, $\errConstr_3=0.3$]{\includegraphics{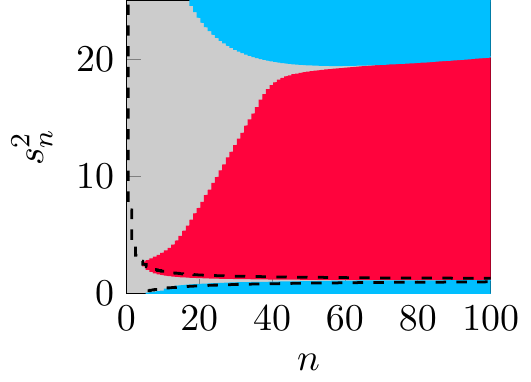}\label{img:RegionsRandomVarMedEstConstr}}
    \hfil
    \subfigure[Relaxed constraint, $\errConstr_3=0.5$]{\includegraphics{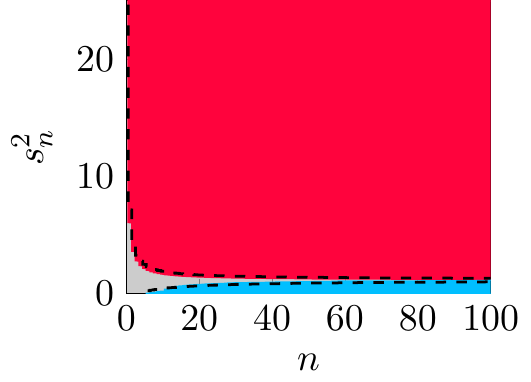}\label{img:RegionsRandomVarWeakEstConstr}}
    \caption{Different regions of the optimal shift-in-mean test. The region $\stopRegionComplWoIdx$ denotes the complement of the stopping region, the region $\stopRegionDecWoIdx{1}$ is the region where the test stops and decides in favor of $\Hyp_1$, the region $\stopRegionDecWoIdx{0}$ is the region where the test stops and decides in favor of $\Hyp_0$. The dashed lines are the thresholds of the \ac{SPRT}.}
    \label{img:RegionsRandomVar}
\end{figure*}

While the stopping and decision regions for the random mean example can be explained in an intuitive way, this is not the case for the shift-in-variance test. By inspecting the different regions of the shift-in-variance test, which are depicted in \cref{img:RegionsRandomVarOrig}, one can see that the stopping region of the test is split into three different regions. Intuitively, one would expect two different regions, one corresponding to a small variance and, hence, a small value of $s_n^2$ and one corresponding to a large variance and, hence, a large value of $s_n^2$.
Since a small variance is more likely for a small value of $s_n^2$ and a small value of $s_n^2$ leads to a sharp posterior distribution, the boundary of the lower blue region in \cref{img:RegionsRandomVarOrig} is close to the Wald threshold.
Also, the complement of the stopping region is located between the lower blue region and the decision region in favor of $\Hyp_1$.
This corridor is close to the Wald thresholds and decreases with an increasing number of samples, as one would expect. For a sufficient statistic $s_n^2>2.5$, there exists another corridor in which the test has to continue due to the estimation uncertainty. The lower boundary (in terms of the sufficient statistic) of the corridor increases with an increasing number of samples since the posterior variance decreases. For a sufficient statistic of $s_n^2>12$, the complement of the stopping region changes to a stopping region in which the test decides in favor of $\Hyp_0$. This is not intuitively clear since the certainty about $\Hyp_0$ should increase with an increasing statistic. Although the certainty of $\Hyp_1$ increases with increasing $s_n^2$, the estimation error also increases.
Our explanation for this effect is as follows: Since we consider a sequential joint detection and estimation problem, it would take very long until the joint cost function would decrease to a sufficient small value. As the aim is to minimize the expected run-length of the scheme while fulfilling constraints on the detection and estimation performance, it is cheaper to make a wrong decision than to wait too long until the estimation uncertainty decreases. Moreover, a truncated sequential scheme is considered, which means that the time which would be needed to reduce to estimation error to an appropriate level is even larger than the maximum number of allowed samples in this case.

\begin{table}[!t]
 \centering
 \caption{Shift-in-variance test: Constraints and simulation results.}
 \subtable[Detection and estimation errors]{
 \label{tbl:ResultsRandomVarErr}
 \includegraphics{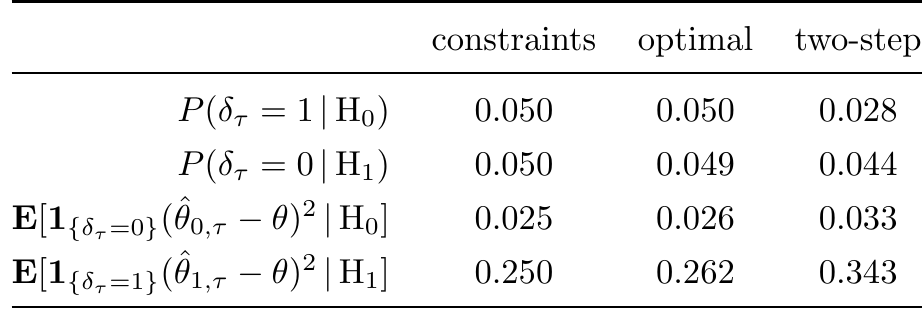}
 }\\
  \subtable[Run-lengths: The second column contains the expected run-length of the optimal scheme obtained as the output of the \ac{LP}.]{
  \label{tbl:ResultsRandomVarRL}
  \includegraphics{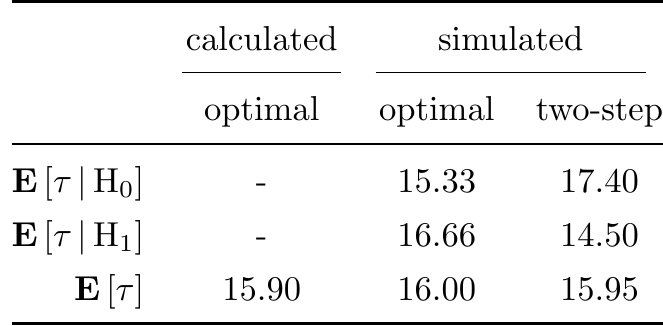}
 }
\end{table}

The empirical results of the optimal scheme as well as for the two-step procedure are summarized in \cref{tbl:ResultsRandomVarErr,tbl:ResultsRandomVarRL}. One can see that the constraints on the error probabilities are---within the range of the Monte Carlo uncertainty---hit exactly by the optimal scheme. For the two-step procedure, the empirical results are, in part, much lower than the constraints. The estimation constraints are almost fulfilled exactly for the optimal scheme, whereas they are exceeded in a non-negligible manner by the benchmark approach. 
Surprisingly, both empirical run-lengths are very close to the run-length resulting from the design procedure. One would expect that the average run-length of the optimal scheme is much smaller than the one of the non-optimal one. This can have two reasons.
First, the relatively small average run-length is due to the fact that the constraints of the \ac{MSE} are violated. Moreover, the optimal as well as the non-optimal schemes are truncated. For the optimal scheme, the maximum number of samples is needed in $0.004\,\%$ of the Monte Carlo runs, which means that $N$ samples are enough for almost all runs. 
On the other hand, in $1.4\,\%$ of the Monte Carlo runs the \ac{SPRT} reached the maximum number of samples. Hence, increasing the maximum number of samples for both, the optimal and the non-optimal scheme would increase the average run-length of the two-step procedure, whereas the average run-length of the optimal test would almost stay the same.

To demonstrate the influence of the estimation constraint $\errConstr_3$ on the region $\stopRegionDecWoIdx{0}$ for large $s_n^2$, the estimation constraint under $\Hyp_1$ is relaxed to  $\errConstr_3=0.3$ in a first and to $\errConstr_3=0.5$ in a second step. The different regions of the two tests with relaxed estimation constraints are depicted in \cref{img:RegionsRandomVarMedEstConstr,img:RegionsRandomVarWeakEstConstr}. One can see that for large $s_n^2$ the region $\stopRegionDecWoIdx{0}$ gets smaller for the first relaxed test while it disappears completely for the second relaxed test. The results of the Monte Carlo simulations for the problems with relaxed estimation constraints can be found in \cref{app:relaxedShiftInVariance}.

Although this problem could be solved numerically in a sufficiently accurate manner using a straightforward discretization of the continuous spaces, the limits of this approach are clearly visible.
Depending on the likelihood, the prior as well as on the maximum number of samples $N$, the posterior distribution $p(\param\given\obsIdx{n})$ and the posterior predictive $p(\xnew\given\obsIdx{n})$ can become very sharp.
Furthermore, the corridor in which the test continues can become very small, especially for large $N$. Sticking to this discretization approach, one would need to deal with very fine grids and, as a result, very high dimensional matrices \updateProbMeasureMtx.

To overcome this issue, different representations of the cost functions have to be used, such as a discretization using non-regular grids or an approximation with basis functions.
 \section{Conclusions}
Based on very mild assumptions, a flexible framework for optimal sequential joint detection and estimation was developed. The resulting tests minimize the average number of samples while fulfilling constraints on the detection and estimation performance.
The resulting cost function was analyzed and a connection between the weights of the individual errors and the performance measures of the test were established.
These results were then used to select a set of optimal weights via Linear Programming, such that all constraints are fulfilled and the resulting scheme is of minimum average run-length.
This automatic selection of the optimal coefficients overcomes the problem that the choice of these coefficients is left to the practitioner in existing approaches.
Although the results are only presented for the two hypotheses case and for scalar random parameters, the proposed scheme can, in principle, be extended to multiple hypotheses and/or multiple random parameters by adopting the loss functions.
To validate the proposed approach, two numerical examples, which are widely used in various applications, were provided. The designed tests were validated using Monte Carlo simulations. Based on these examples, the behavior of the sequential test for different detection/estimation error constraints is shown. By comparing the optimal sequential scheme with an \ac{SPRT} followed by an \ac{MMSE} estimator, the gap in performance when using sub-optimal schemes was shown.

Though for more complex scenarios, like hierarchical Bayes models \citep{Makni2008} or high-dimensional problems,  a straightforward discretization of the problem may not be sufficient anymore and more advanced numerical methods may have to be used to represent the densities, the theoretical results of this work remain valid.
 
\appendix
\crefalias{section}{appsec}
 \section{Proof of Theorem \ref{theo:solutionOptimalStopping}}\label{app:proofSolutionOptimalStopping}
The proof of \cref{theo:solutionOptimalStopping} follows from the fundamental results of optimal stopping theory. We follow here the proofs stated in \citet{fauss2015linear,novikov2009optimal,poor2009quickest}. 
In this proof, a truncated optimal stopping problem with finite horizon $N\geq1$ is considered with the cost $c_n$ for stopping at time $n$. Let 
\begin{align*}
 V_n = \min\{c_n, \E\left[V_{n+1}\given\stat{n}\right]\}\,,\quad n<N
\end{align*}
be the cost for stopping at the optimal time instant between $n$ and $N$ with basis $V_N=c_N$. The minimal cost is defined by
\begin{align*}
 V^\star = V_0\,.
\end{align*}
The instantaneous cost $c_n$ is obtained from \cref{eq:optStoppingFormulation} and is then given by
\begin{align*}
 c_n = n + \costS(\stat{n})\,.
\end{align*}
The following proof is performed by induction. Assume that for some $n<N$ the optimal cost function has the form $ V_n = n + \rho_n(\stat{n})$. It has to be shown that this relation holds also for $n-1$ if it holds for $n$.
The cost function at time $n-1$ can be written as
\begin{align*}
 V_{n-1} & = \min\left\{c_{n-1}, \E\left[V_{n}\given\stat{n-1}\right]\right\} \\
	 & = \min\left\{n-1 + \costS(\stat{n-1}), \E\left[n+\rho_n(\stat{n})\given\stat{n-1}\right]\right\} \\
	 & = n-1 + \min\left\{\costS(\stat{n-1}), 1+\E\left[\rho_n(\stat{n})\given\stat{n-1}\right]\right\} \\
	 & = n-1 + \rho_{n-1}(\stat{n-1})\,,
\end{align*}
where the expected value
\begin{align*}
 \E\left[\rho_n(\stat{n})\given\stat{n-1}\right] & = 1 + \int \rho_n(\stat{n}) p(\obsScalar{n}\given\stat{n-1}) \dd\obsScalar{n} \\
						 & = 1 + \int \rho_n(\xi(\stat{n-1},\obsScalar{n})) p(\obsScalar{n}\given\stat{n-1}) \dd\obsScalar{n}
\end{align*}
is a function of $\stat{n-1}$ only. Thus, the induction step holds.
The induction basis is given by $ V_N = N + g(\stat{N})$.
Hence, the minimum cost of the truncated optimal stopping problem is given by $V^\star = V_0=\rho_0(\stat{0})$.

\section{Proof of Corollary \ref{corr:H-integrable}}\label{app:proof_corr_H-integrable}
In what follows, it is shown that $\rho_{n+1}$ is \updateProbMeasure-integrable for all $\stat{n}$ and for all $0\leq n < N$. From the definition of $\rho_n$, one can directly see that
\begin{align*}
 \int\rho_{n+1}\dd\updateProbMeasure \leq \int g \dd\updateProbMeasure \leq \int  \auxVarCostOpt{i,n+1}\dd\updateProbMeasure\,.
\end{align*}
By definition of $\auxVarCostOpt{i,n}$, $i\in\{0,1\}$, one obtains
\begin{align*}
 \int\auxVarCostOpt{i,n+1}\dd\updateProbMeasure = & \int \left(C_{1-i}p(\Hyp_{1-i}\given\xi(\stat{n},\xnew)) \, + C_{2+i}p(\Hyp_i\given \xi(\stat{n},\xnew))\,\Var\bigl[\paramRV\given\xi(\stat{n},\xnew),\Hyp_i\bigr]\right)p(\xnew\given\stat{n})\dd\xnew \\
 = & \int C_{1-i}p(\Hyp_{1-i}\given\stat{n},\xnew)p(\xnew\given\stat{n}) \, + C_{2+i}p(\Hyp_i\given \stat{n},\xnew)\,\Var\bigl[\paramRV\given\stat{n},\xnew,\Hyp_i\bigr]p(\xnew\given\stat{n})\dd\xnew\,.
\end{align*}
The first term of the integral on the right hand side of the previous equation can be expressed as
\begin{align*}
 \int C_{1-i}p(H_{1-i}\given\stat{n},\xnew) p(\xnew\given\stat{n})\dd\xnew = \int C_{1-i}p(H_{1-i},\xnew\given\stat{n})\dd\xnew = C_{1-i}p(H_{1-i}\given\stat{n})\,,
\end{align*}
whereas the second term can be written as
\begin{align*}
 & \int C_{2+i}p(\Hyp_i\given \stat{n},\xnew)\,\Var\bigl[\paramRV\given\stat{n},\xnew,\Hyp_i\bigr]p(\xnew\given\stat{n})\dd\xnew = C_{2+i} p (\Hyp_i\given\stat{n}) \int \Var\bigl[\paramRV\given\stat{n},\xnew,\Hyp_i\bigr] p(\xnew\given \stat{n},\Hyp_i) \dd \xnew\,.
\end{align*}
From now on, the following notation is used
\begin{align*}
 \E_{\xnew \sim p(\xnew\given \stat{n}, \Hyp_i)} \left[ \Var\bigl[\paramRV\given\stat{n},\xnew,\Hyp_i \bigr]\right] := \int \Var\bigl[\paramRV\given\stat{n},\xnew,\Hyp_i\bigr] p(\xnew\given \stat{n},\Hyp_i) \dd \xnew\,.
\end{align*}
To calculate the conditional expectation of the posterior variance, the conditional expectation of the posterior mean has to be calculated first. It is given by
\begin{align*}
 \bar{\param} := & \E_{\xnew \sim p(\xnew\given \stat{n}, \Hyp_i)} \left[ \E\bigl[\paramRV\given\stat{n},\xnew,\Hyp_i \bigr]\right] \\
	       = & \int \int \param p(\param \given \stat{n}, \xnew, \Hyp_i ) \dd \param p(\xnew\given \stat{n}, \Hyp_i) \dd \xnew \\
	       = & \int \int \param p(\param , \xnew \given \stat{n},  \Hyp_i )  \dd \xnew \dd \param = \E\bigl[ \param \given \stat{n}, \Hyp_i \bigr]\,.
\end{align*}
The conditional expectation of the posterior variance is hence given by
\begin{align*}
  \E_{\xnew \sim p(\xnew\given \stat{n}, \Hyp_i)} \left[ \Var\bigl[\paramRV\given\stat{n},\xnew,\Hyp_i \bigr]\right] = & \int \int \left( \param - \bar{\param} \right )^2 p(\param \given \stat{n},\xnew,\Hyp_i) \dd\theta p(\xnew \given \stat{n},\Hyp_i) \dd\xnew \\
  = & \int \int \left( \param - \bar{\param} \right)^2 p(\param,\xnew\given\stat{n},\Hyp_i) \dd\xnew\dd\param\\
  = &  \int \left( \param - \bar{\param}\right)^2 p(\param\given\stat{n},\Hyp_i) \dd\param =  \Var\bigl[ \paramRV \given \stat{n}, \Hyp_i \bigr]\,.
\end{align*}
With these results, one can conclude that
\begin{align*}
 \int\auxVarCostOpt{i,n+1}\dd\updateProbMeasure = & C_{1-i}p(\Hyp_{1-i}\given\stat{n}) + C_{2+i}p(\Hyp_i\given\stat{n})\Var\bigl[\paramRV\given\stat{n},\Hyp_i\bigr] < \infty\,,
\end{align*}
since the constants are assumed to be finite and the posterior variance is finite, following Assumption \ref{ass:finiteMoments}.
Therefore, the integral of the cost function is bounded by
\begin{align*}
 \int\rho_{n+1}\dd\updateProbMeasure \leq \int g \dd\updateProbMeasure \leq  \int\auxVarCostOpt{i,n+1}\dd\updateProbMeasure = C_{1-i}p(\Hyp_{1-i}\given\stat{n}) + C_{2+i}p(\Hyp_i\given\stat{n})\Var\bigl[\paramRV\given\stat{n},\Hyp_i\bigr] < \infty\,,
\end{align*}
so that $\rho_{n+1}$ is \updateProbMeasure-integrable for all $\stat{n}$ and all $0\leq n < N$.

\section{Proof of Lemma \ref{lem:rhoBounded}} \label{app:proofLemRhoBounded}
The proof of \cref{lem:rhoBounded} is only outlined for the upper bound, since the lower bound can be shown analogously.
First of all, it has to be mentioned that the function relating $\postProbwo_{n+1}$ and $\postProbwo_n$ is linear in $\postProbwo_n$, i.e.,
\begin{align}\label{eq:transPostProbLin}
 \postProb[n+1]{i} = \transkernelPostVar{\postProb{i}}{x_{n+1}}{\stat{n}}  & =  \postProb{i} \frac{p(x_{n+1}\given\stat{n},\Hyp_i)}{p(x_{n+1}\given\stat{n})}\,.
\end{align}
The following proof is performed by induction. Let $a^\circ=\max\{a_0,a_1,1\}$ and assume that \cref{lem:rhoBounded} holds for some $n<N$. For $n-1$, by applying \cref{eq:transPostProbLin} it holds that
\begin{align*}
 \tilde \rho_{n-1}(a\cdot\postProbwo_{n-1}, \stat{n-1}) & =  \min\left\{ \tilde g(\stat{n-1}, a\cdot\postProbwo_{n-1}), 1 + \!\int \tilde\rho_{n}\bigl(\transkernelPostVar{a\cdot \postProbwo_{n-1}}{\xnew}{\stat{n-1}}, \xi_{\stat{n-1}}(\xnew) \bigr) p(\xnew\given\stat{n-1}) \dd \xnew \right\} \\
 & = \min\left\{ \tilde g(\stat{n-1}, a\cdot\postProbwo_{n-1}), 1 +\!\!\int\!\tilde\rho_{n}\bigl(a\cdot\transkernelPostVar{\postProbwo_{n-1}}{\xnew}{\stat{n-1}}, \xi_{\stat{n-1}}(\xnew) \bigr) p(\xnew\given\stat{n-1}) \dd \xnew \right\}\,.
 \end{align*}
 Applying \cref{lem:gBounded,lem:rhoBounded} yields
 \begin{align*}
  \tilde \rho_{n-1}(a\cdot\postProbwo_{n-1}, \stat{n-1}) & \leq \min\left\{ a^\circ \tilde g(\stat{n-1}, \postProbwo_{n-1}), 1 + \int \tilde\rho_{n}\bigl(a\transkernelPostVar{\postProbwo_{n-1}}{\xnew}{\stat{n-1}}, \xi_{\stat{n-1}}(\xnew) \bigr) p(\xnew\given\stat{n-1}) \dd \xnew \right\} \\
  & \leq \min\left\{ a^\circ \tilde g(\stat{n-1}, \postProbwo_{n-1}), 1 + a^\circ\!\int\! \tilde\rho_{n}\bigl(\transkernelPostVar{\postProbwo_{n-1}}{\xnew}{\stat{n-1}}, \xi_{\stat{n-1}}(\xnew) \bigr) p(\xnew\given\stat{n-1}) \dd \xnew \right\}\,.
  \end{align*}
We can further state that
  \begin{align*}
  & \phantom{\leq}\min\left\{ a^\circ \tilde g(\stat{n-1}, \postProbwo_{n-1}), 1 + a^\circ\!\int\! \tilde\rho_{n}\bigl(\transkernelPostVar{\postProbwo_{n-1}}{\xnew}{\stat{n-1}}, \xi_{\stat{n-1}}(\xnew) \bigr) p(\xnew\given\stat{n-1}) \dd \xnew \right\} \\
   &  \leq  \min\left\{ a^\circ \tilde g(\stat{n-1}, \postProbwo_{n-1}), a^\circ + a^\circ\!\!\! \int\!\! \tilde\rho_{n}\bigl(\transkernelPostVar{\postProbwo_{n-1}}{\xnew}{\stat{n-1}}, \xi_{\stat{n-1}}(\xnew) \bigr) p(\xnew\given\stat{n-1}) \dd \xnew \right\} \\
   &  =  a^\circ\min\left\{ \tilde g(\stat{n-1}, \postProbwo_{n-1}), 1 +  \!\int\!\tilde\rho_{n}\bigl(\transkernelPostVar{\postProbwo_{n-1}}{\xnew}{\stat{n-1}}, \xi_{\stat{n-1}}(\xnew) \bigr) p(\xnew\given\stat{n-1}) \dd \xnew \right\}\,.
\end{align*}
Based on this, we can conclude that
\begin{align*}
  \tilde \rho_{n-1}(a\cdot\postProbwo_{n-1}, \stat{n-1}) \leq a^\circ\min\left\{ \tilde g(\stat{n-1}, \postProbwo_{n-1}), 1 +  \!\int\!\tilde\rho_{n}\bigl(\transkernelPostVar{\postProbwo_{n-1}}{\xnew}{\stat{n-1}}, \xi_{\stat{n-1}}(\xnew) \bigr) p(\xnew\given\stat{n-1}) \dd \xnew \right\}\,.
\end{align*}
 The induction basis is given by $n=N$, where it holds that
\begin{align*}
 \tilde \rho_{N}(\stat{N},a\cdot\postProbwo_{N}) = \tilde g(\stat{N},a\postProbwo_N) \leq a^\circ \tilde g(\stat{N},\postProbwo_N) =  a^\circ \tilde \rho_{N}(\stat{N}, \postProbwo_N)\,
\end{align*}
due to \cref{lem:gBounded}}. This concludes the proof.

\section{Proof of Theorem \ref{theo:costDerivatives}}\label{app:proofCostDerivatives}
Let
\begin{align*}
 \rho^\prime_{n,C_i}(\stat{n}) = & \frac{\partial \rho_{n}(\stat{n})}{\partial C_i}\\
		     = & \left\{
		     \begin{array}{ll}
                     g^\prime_{C_i}(\stat{n}) & \text{ for } \stat{n} \in \stopRegion \\
                     \frac{\partial}{\partial C_i} \int \rho_{n+1}\dd\updateProbMeasure   & \text{ for }\stat{n} \in \stopRegionCompl
		     \end{array}
		     \right.
\end{align*}
denote the derivatives of the cost function with respect to $C_i$, which are defined everywhere on $E_\statWOa \setminus \stopRegionBound$.
Assume that the order of integration and differentiation can be interchanged, i.e.,
\begin{align}
 \frac{\partial }{\partial C_i} \int \rho_{n+1}\dd\updateProbMeasure & = \int \rho^\prime_{n+1,C_i}\dd\updateProbMeasure \label{eq:interchDiffInt}\,.
\end{align}
The validity of this assumption is shown later on. In general, the derivative of the cost function can be expressed as
\begin{align}
  \rho^\prime_{n,C_i}(\stat{n}) = \indd{\stopRegion} g^\prime_{C_i}(\stat{n}) + \indd{\stopRegionCompl} \int \rho^\prime_{n+1,C_i}\dd\updateProbMeasure \label{eq:derivateRho}\,.
\end{align}
First, the derivatives with respect to the coefficients corresponding to the detection, $C_{\{0,1\}}$,  are calculated.
The derivative of the cost function on $\stopRegion$ given by
\begin{align}
 g^\prime_{n,C_i}(\stat{n}) = & \indd{\stopRegionDec{n}{1-i}}p(\Hyp_i\given\stat{n}) \nonumber \\
		  = & \indd{\stopRegionDec{n}{1-i}} p(\Hyp_i)\frac{p(\stat{n}\given\Hyp_i)}{p(\stat{n})} \,. \label{eq:derCostStopDet}
\end{align}
By using the property $p(\stat{n}\given\stat{n-1}) = p(x_n\given\stat{n-1})$ and defining the short-hand notation
\begin{align}\label{eq:appShortHandz}
 z^i_{n} = \frac{p(\stat{n}\given\Hyp_i)}{p(\stat{n})}\;,
\end{align}
\cref{eq:derCostStopDet} can be written as
\begin{align}
 g^\prime_{n,C_i}(\stat{n}) = & \indd{\stopRegionDec{n}{1-i}} p(\Hyp_i)\frac{p(\stat{n}\given \stat{n-1},\Hyp_i)p(\stat{n-1}\given\Hyp_i) }	{p(\stat{n-1})p(\stat{n}\given\stat{n-1})}\nonumber\\
 = & \indd{\stopRegionDec{n}{1-i}} p(\Hyp_i)z^i_{n-1}\frac{p(\obsScalar{n}\given \stat{n-1},\Hyp_i) }{p(\obsScalar{n}\given\stat{n-1})} \label{eq:derCostStopDetFinal}\,.
\end{align}
Combining \cref{eq:derivateRho} and \cref{eq:derCostStopDetFinal} leads to:
\begin{align*}
  \rho^\prime_{n,C_i}(\stat{n}) & = \indd{\stopRegionDec{n}{1-i}}  p(\Hyp_i)\frac{p(\stat{n}\given\Hyp_i)}{p(\stat{n})} \\
		        & + \indd{\stopRegionCompl[n]} \biggl(\int_{\{\xi_{\stat{n}} \in \stopRegionDec{n+1}{1-i}\}}  p(\Hyp_i) z^i_{n}\frac{p(\xnew\given \stat{n},\Hyp_i) }{p(\xnew\given\stat{n})}  p(\xnew\given\stat{n}) \dd\xnew + \int_{\stopRegionCompl[n+1]} \rho_{n+1,C_i}^\prime  \dd\updateProbMeasure \biggr) \\
		        & = \indd{\stopRegionDec{n}{1-i}}  p(\Hyp_i)z_n^i + \indd{\stopRegionCompl[n]} \biggl( p(\Hyp_i) z_n^i \updateProbMeasure^i\bigl(\stopRegionDec{n+1}{1-i}\bigr)  + \int_{\stopRegionCompl[n+1]} \rho_{n+1,C_i}^\prime  \dd\updateProbMeasure \biggr)
\end{align*}
This coincides with the results for $i\in\{0,1\}$ stated in \cref{theo:costDerivatives}.
Second, the derivatives with respect to the coefficients corresponding to the estimation, i.e., $C_{\{2,3\}}$, have to be calculated.
Using the short hand notation defined in \cref{eq:appShortHandz}, the derivative on $\stopRegion$ is given by
\begin{align}
 g^\prime_{n,C_i}(\stat{n}) = & \indd{\stopRegionDec{n}{i-2}} p(\Hyp_{i-2})z^{i-2}_{n}\Var\left[\paramRV\given\stat{n},\Hyp_{i-2}\right]\,.\label{eq:appGprimeEst}
\end{align}
The derivative on $\stopRegionCompl$ can be expressed as
\begin{align}
 \rho^\prime_{n,C_i}(\stat{n}) =  & \int \rho^\prime_{n+1,C_i}\dd\updateProbMeasure \nonumber\\
			       =  & \int_{\stopRegion[n+1]} g^\prime_{n+1,C_i}\dd\updateProbMeasure + \int_{\stopRegionCompl[n+1]} \rho_{n+1,C_i}^\prime  \dd\updateProbMeasure\,.\label{eq:appRhoPrimeComplEst}
\end{align}
Following the same lines as for $i\in\{0,1\}$, the first integral in the last line can be rewritten as
\begin{align}
 \int_{\stopRegion[n+1]} g^\prime_{n+1,C_i}\dd\updateProbMeasure = & \int_{\{\xi_{\stat{n}} \in \stopRegionDec{n+1}{i-2}\}} p(\Hyp_{i-2})z^{i-2}_{n+1}\Var\left[\paramRV\given\stat{n+1},\Hyp_{i-2}\right] p(\xnew\given\stat{n})\dd\xnew\nonumber\\
 = & z^{i-2}_{n} p(\Hyp_{i-2}) \int_{\{\xi_{\stat{n}} \in \stopRegionDec{n+1}{i-2}\}} \Var\left[\paramRV\given\xi(\stat{n},\xnew),\Hyp_{i-2}\right] \frac{p(\xnew\given\stat{n},\Hyp_{i-2})}{p(\xnew\given\stat{n})} p(\xnew\given\stat{n})\dd\xnew \nonumber\\
  = & z^{i-2}_{n} p(\Hyp_{i-2}) \int_{\{\xi_{\stat{n}} \in \stopRegionDec{n+1}{i-2}\}} \Var\left[\paramRV\given\xi(\stat{n},\xnew),\Hyp_{i-2}\right] p(\xnew\given\stat{n},\Hyp_{i-2})\dd\xnew\,.\label{eq:appRhoPrimeComplEst2}
\end{align}
Combining \cref{{eq:appGprimeEst},{eq:appRhoPrimeComplEst},{eq:appRhoPrimeComplEst2}}, one yields
\begin{align*}
  \rho^\prime_{n,C_i}(\stat{n}) & = \indd{\stopRegionDec{n}{i-2}} p(\Hyp_{i-2})z^{i-2}_{n}\Var\left[\paramRV\given\stat{n},\Hyp_{i-2}\right] \\
		        & + \indd{\stopRegionCompl[n]} \biggl( p(\Hyp_{i-2})z^{i-2}_{n} \int_{\{\xi_{\stat{n}} \in \stopRegionDec{n+1}{i-2}\}} \Var\left[\paramRV\given\xi(\stat{n},\xnew),\Hyp_{i-2}\right]  p(\xnew\given\stat{n},\Hyp_{i-2}) \dd\xnew\\
		        & + \int_{\stopRegionCompl[n+1]} \rho_{n+1,C_i}^\prime  \dd\updateProbMeasure \biggr)\,\\
\end{align*}
which is the result stated in \cref{theo:costDerivatives} for $i\in\{2,3\}$.

It now has to be shown that \cref{eq:interchDiffInt} holds, i.e., that the order of the integral and the differentiation can be interchanged. According to the differentiation lemma \citep[Lemma 16.2.]{bauer2001measure}, the order can be interchanged if the following conditions hold:
\begin{enumerate}
 \item The function $\rho_{n+1}(\stat{n+1})$ has to be \updateProbMeasure-integrable for all $0\leq n < N$ and all $C_i$, $i\in\{0,1,2,3\}$.
 \item The function $\rho_n(\stat{n})$ is differentiable for all $0\leq n \leq N$ and all $\stat{n}\in E_\statWOa$.
 \item A function $h_n^i(\stat{n})$ has to exist, which is independent of $C_i$ and it must further hold that
 \begin{align*}
  \given \rho^\prime_{n,C_i}(\stat{n})\given \leq h_n^i(\stat{n}),\quad i\in\{0,1,2,3\}\,.
 \end{align*}
\end{enumerate}
Condition 1 is fulfilled by \cref{corr:H-integrable}, but conditions 2 and 3 still have to be proven. The proof is only carried out for $i\in\{0,1\}$, since it can be done analogously for $i\in\{2,3\}$.
Suppose that the derivative lemma holds for some $\rho_n$ and some $i\in\{0,1\}$. Further assume that
\begin{align*}
 r^i_n(\stat{n}) = p(\Hyp_i) z_n^i \updateProbMeasure^i\bigl(\stopRegionDec{n+1}{1-i}\bigr) + \int_{\stopRegionCompl[n+1]} \rho_{n+1,C_i}^\prime  \dd\updateProbMeasure
\end{align*}
is its derivative on $\stopRegionCompl$. Now consider the derivative of $\rho_{n-1}$, is well-defined and bounded on $\stopRegion$. On $\stopRegionCompl$, it holds that
\begin{align}\label{eq:derLemmaCond2}
 \frac{\partial \rho_{n-1}}{\partial C_i} = \frac{\partial }{\partial C_i}\int\rho_{n}\dd\updateProbMeasure[n-1] & = p(\Hyp_i) z_{n-1}^i \updateProbMeasure[n-1]^i\bigl(\stopRegionDec{n}{1-i}\bigr) + \int_{\stopRegionCompl[n]} \rho_{n,C_i}^\prime  \dd\updateProbMeasure[n-1]\,.
\end{align}
Hence, the derivative of $\rho_{n-1}$ is upper bounded on $\stopRegionCompl$ by
\begin{align}
 \left\lvert \frac{\partial \rho_{n-1}}{\partial C_i} \right\rvert &= \left\lvert p(\Hyp_i) z_{n-1}^i \updateProbMeasure[n-1]^i\bigl(\stopRegionDec{n}{1-i}\bigr) + \int_{\stopRegionCompl[n]} \rho_{n,C_i}^\prime  \dd\updateProbMeasure[n-1] \right\rvert \nonumber\\
 & \leq  p(\Hyp_i) z_{n-1}^i \updateProbMeasure[n-1]^i\bigl(\stopRegionDec{n}{1-i}\bigr) + \left\lvert\int_{\stopRegionCompl[n]} \rho_{n,C_i}^\prime  \dd\updateProbMeasure[n-1] \right\rvert \nonumber\\
  & \leq  p(\Hyp_i) z_{n-1}^i \updateProbMeasure[n-1]^i\bigl(\stopRegionDec{n}{1-i}\bigr) + \int_{\stopRegionCompl[n]} h_n^i  \dd\updateProbMeasure[n-1] \label{eq:derLemmaCond3}\,.
\end{align}
From \cref{eq:derLemmaCond2,eq:derLemmaCond3}, it follows that the derivative Lemma holds for $n-1$ if it holds for $n$.
The induction basis is given by $\rho_N$, where it holds that
\begin{align*}
 \rho^\prime_{N,C_i} = p(\Hyp_i)z_N^i\ind{\stopRegionDec{n}{1-i}}
\end{align*}
and
\begin{align*}
 \left\lvert \rho^\prime_{N,C_i} \right\rvert \leq p(\Hyp_i)z_N^i = h_N^i(\stat{n})\,.
\end{align*}
This concludes the proof.

\section{Proof of Theorem \ref{theo:derivativesPerformanceRelation}}\label{app:proofDerivativesPerformanceRelation}
Assuming that the optimal policy stated in \cref{cor:optPolicyC} is used, the performance measures defined in \cref{eq:alphaRec} and \cref{eq:betaRec} are written in terms of the different regions of the test which are defined in \cref{eq:defRegionsn,eq:defRegionsN,eq:defStopRegionDec}.

The error probabilities $\errorDet{n}{i}(\stat{n})$ of the scheme at time $n$ and state $\stat{n}$ are then given by:
\begin{align}\label{eq:alphaDefRegions}
 \errorDet{n}{i}(\stat{n}) & = \left\{
		     \begin{array}{ll}
                      0 & \!  \text{for } \stat{n} \in \stopRegionDec{n}{i} \\
                      1 &\!  \text{for } \stat{n} \in \stopRegionDec{n}{1-i} \\
                      \E[\errorDet{n+1}{i}(\stat{n+1})\given\stat{n},\Hyp_i] & \!  \text{for }\stat{n} \in \stopRegionCompl
		     \end{array}
		     \right.
\end{align}
The expected value in \cref{eq:alphaDefRegions} is equivalent to
\begin{align*}
 \E\left[\errorDet{n+1}{i}(\stat{n+1})\given\stat{n},\Hyp_i\right] & = \int \errorDet{n+1}{i} \dd\updateProbMeasure  \\
		 & = \updateProbMeasure^i\bigl(\stopRegionDec{n+1}{1-i}\bigr) + \int_{\stopRegionCompl[n+1]} \errorDet{n+1}{i}  \dd\updateProbMeasure^i
\end{align*}
To show that \cref{theo:derivativesPerformanceRelation} holds, it first has  to be shown that $\errorDet{n}{i}(\stat{n})=\frac{\rho^\prime_{n,C_i}(\stat{n})}{p(\Hyp_i)z_n^i}$ is a valid solution of \cref{eq:alphaDefRegions}.
For $\stat{n}\in\stopRegionDec{n}{i}$, it holds that
\begin{align*}
 \frac{\rho^\prime_{n,C_i}}{p(\Hyp_i)z_n^i} = 0 = \errorDet{n}{i}(\stat{n})\,
\end{align*}
and for $\stat{n}\in\stopRegionDec{n}{1-i}$ it holds that
\begin{align*}
 \frac{\rho^\prime_{n,C_i}}{p(\Hyp_i)z_n^i} = \frac{ p(\Hyp_i) z_n^i}{p(\Hyp_i)z_n^i} = 1 = \errorDet{n}{i}(\stat{n})\,,
\end{align*}
which corresponds to the error probabilities on the stopping region.
On the complement of the stopping region, we have
\begin{align*}
 \frac{\rho^\prime_{n,C_i}(\stat{n})}{p(\Hyp_i)z_n^i} & =  \updateProbMeasure^i\bigl(\stopRegionDec{n+1}{1-i}\bigr)  + \int_{\{\xi_{\stat{n}}\in\stopRegionCompl[n+1]\}} \frac{\rho_{n+1,C_i}^\prime(\xi(\stat{n},\xnew))}{p(\Hyp_i)z_n^i\frac{p(\xnew\given \stat{n},\Hyp_i)}{p(\xnew\given\stat{n})}} p(\xnew\given\stat{n}, \Hyp_i) \dd\xnew \\
 \frac{\rho^\prime_{n,C_i}(\stat{n})}{p(\Hyp_i)z_n^i} & =  \updateProbMeasure^i\bigl(\stopRegionDec{n+1}{1-i}\bigr) + \frac{1}{p(\Hyp_i)z_n^i}\int_{\stopRegionCompl[n+1]}\rho_{n+1,C_i}^\prime \dd\updateProbMeasure \\
\rho^\prime_{n,C_i}(\stat{n}) & =  p(\Hyp_i)z_n^i \updateProbMeasure^i\bigl(\stopRegionDec{n+1}{1-i}\bigr) + \int_{\stopRegionCompl[n+1]}\rho_{n+1,C_i}^\prime \dd\updateProbMeasure
\end{align*}
which is true by \cref{theo:costDerivatives}. Hence, \cref{theo:derivativesPerformanceRelation} holds for the error probabilities $\errorDet{n}{i}(\stat{n})$.

In a similar way, the mean squared error can be written as
\begin{align}\label{eq:betaDefRegions}
  \errorEst{n}{i}(\stat{n}) & = \left\{
		     \begin{array}{ll}
                      0 & \text{for } \stat{n} \in \stopRegionDec{n}{i} \\
                      \Var[\paramRV\given\stat{n},\Hyp_i] & \text{for } \stat{n} \in \stopRegionDec{n}{1-i} \\
                      \E[\errorEst{n+1}{i}(\stat{n+1})\given \stat{n},\Hyp_i] & \text{for }\stat{n} \in \stopRegionCompl
		     \end{array}
		     \right.
\end{align}
where the expected value can be expressed as
\begin{align*}
 \E\bigl[\errorEst{n+1}{i} & (\stat{n+1})\given\stat{n},\Hyp_i\bigr] =   \int \errorEst{n+1}{i}\dd\updateProbMeasure  \\
		 & = \int_{\{\xi_{\stat{n}} \in \stopRegionDec{n+1}{i}\}} \Var\bigl[\paramRV\given\xi_{\stat{n}}(\xnew),\Hyp_{i}\bigr]  p(\xnew\given\stat{n},\Hyp_{i}) \dd\xnew +\int_{\stopRegionCompl[n+1]} \errorEst{n+1}{i} \dd\updateProbMeasure^i\,.
\end{align*}
It now has to be shown that $\errorEst{n}{i}(\stat{n})=\frac{\rho^\prime_{n,C_j}(\stat{n})}{p(H_{j-2}) z_n^{j-2}}$  also solves \cref{eq:betaDefRegions}.
For $\stat{n}\in\stopRegionDec{n}{i}$ and with the substitution $i=j-2$, it holds that
\begin{align*}
 \frac{\rho^\prime_{n,C_j}(\stat{n})}{p(H_{j-2}) z_n^{j-2}} & =  0 = \errorEst{n}{i}(\stat{n})\,
\end{align*}
and for $\stat{n}\in\stopRegionDec{n}{1-i}$ with $i=j-2$ it holds that
\begin{align*}
 \frac{\rho^\prime_{n,C_j}(\stat{n})}{p(H_{j-2}) z_n^{j-2}} =  \Var[\paramRV\given\stat{n},\Hyp_{j-2}] = \Var[\paramRV\given\stat{n},\Hyp_i] = \errorEst{n}{i}(\stat{n}) \,,
\end{align*}
which is the definition of the estimation error on the stopping region. For the complement of the stopping region, it also has to be shown that $\errorEst{n}{i}(\stat{n})=\frac{\rho^\prime_{n,C_j}(\stat{n})}{p(H_{j-2}) z_n^{j-2}}$ solves \cref{eq:betaDefRegions}:
\begin{align*}
 \frac{\rho^\prime_{n,C_j}(\stat{n})}{p(H_{j-2}) z_n^{j-2}} & = \int_{\{\xi_{\stat{n}} \in \stopRegionDec{n+1}{j-2}\}}\Var\bigl[\paramRV\given\xi(\stat{n},\xnew),\Hyp_{j-2}\bigr]  p(\xnew\given\stat{n},\Hyp_{j-2}) \dd\xnew\\
	  & + \int_{\{\xi_{\stat{n}}\in\stopRegionCompl[n+1]\}} \frac{\rho_{n+1,C_j}^\prime(\xi(\stat{n},\xnew))}{p(\Hyp_{j-2})z_n^{j-2}\frac{p(\xnew\given \stat{n},\Hyp_{j-2})}{p(\xnew\given\stat{n})}} p(\xnew\given\stat{n}, \Hyp_{j-2}) \dd\xnew \\
 \frac{\rho^\prime_{n,C_j}(\stat{n})}{p(H_{j-2}) z_n^{j-2}} & = \int_{\{\xi_{\stat{n}} \in \stopRegionDec{n+1}{j-2}\}} \Var\bigl[\paramRV\given\xi(\stat{n},\xnew),\Hyp_{j-2}\bigr]  p(\xnew\given\stat{n},\Hyp_{j-2}) \dd\xnew\\
	  & + \frac{1}{p(\Hyp_{j-2})z_n^{j-2}} \int_{\stopRegionCompl[n+1]} \rho_{n+1,C_j}^\prime\dd\updateProbMeasure \\
 \rho^\prime_{n,C_j}(\stat{n}) & = p(H_{j-2}) z_n^{j-2} \int_{\{\xi_{\stat{n}} \in \stopRegionDec{n+1}{j-2}\}} \Var\bigl[\paramRV\given\xi(\stat{n},\xnew),\Hyp_{j-2}\bigr]  p(\xnew\given\stat{n},\Hyp_{j-2}) \dd\xnew\\
	  & + \int_{\stopRegionCompl[n+1]} \rho_{n+1,C_j}^\prime\dd\updateProbMeasure
\end{align*}
which is true by \cref{theo:costDerivatives}. Hence, \cref{theo:derivativesPerformanceRelation} holds also for the estimation error $\errorEst{n}{i}(\stat{n})$.

Finally, the detection and estimation error at $n=0$ with sufficient statistic $\stat{0}$ can be expressed as
\begin{align*}
  \rho^\prime_{0,C_i}(\stat{0}) & = p(\Hyp_i) \errorDetOpt{0}{i}(\stat{0}) \quad\quad\;\;\; i\in\{0,1\}\,, \\
  \rho^\prime_{0,C_i}(\stat{0}) & = p(\Hyp_{i-2}) \errorEstOpt{0}{i-2}(\stat{0}) \quad  i\in\{2,3\}\,,
\end{align*}
which concludes the proof.
\section{Proof of Theorem \ref{theo:maxC}}\label{app:proofTheoMaxC}
It has to be shown that for $C=\{C_0,C_1,C_2,C_3\}$ the solution of
\begin{align}\label{eq:appLagrangianDual}
 \max_{C\geq0}\, L_\errConstr(C) = \max_{C\geq0}\, \rho_0(\stat{0}) - \sum_{i=0}^1 p(\Hyp_i)C_i\errConstr_i - \sum_{i=2}^3 p(\Hyp_{i-2})C_i\errConstr_i
\end{align}
coincides with the solution of \cref{prb:constrProblem}.
According to \cref{theo:derivativesPerformanceRelation}, it holds for the derivative of $L_\errConstr(C)$ with respect to $C_i$ that
\begin{align*}
 \frac{\partial L_\errConstr(C)}{\partial C_i} & = \frac{\partial \rho_0(\stat{0})}{\partial C_i} - p(\Hyp_i)\errConstr_i = p(\Hyp_i)\left(\errorDetOpt{0}{i}(\stat{0};\policyOptC) - \errConstr_i\right)\,i\in\{0,1\}\,,\\
 \frac{\partial L_\errConstr(C)}{\partial C_i} & = \frac{\partial \rho_0(\stat{0})}{\partial C_i} - p(\Hyp_{i-2})\errConstr_i = p(\Hyp_i)\left(\errorEstOpt{0}{i}(\stat{0};\policyOptC) - \errConstr_i\right)\,i\in\{2,3\}\,.
\end{align*}
Due to space constraints, the proof of optimality is only shown for the detection constraints because it follows analogously for the estimation constraints.
Since the constraints on the detection and estimation errors are inequality constraints, $C_i^\star$ is an optimal solution of \cref{eq:appLagrangianDual} if the derivative vanishes for a non-negative $C_i^\star$ or there exist a non-positive derivative for $C_i^\star=0$.
First of all, assume that there exist a non-negative and bounded $C_i^\star$ for which die derivative vanishes. In this case, it holds that
\begin{align*}
 \errorDetOpt{0}{i}(\stat{0};\policyOptCerr) = \errConstr_i\,,
\end{align*}
i.e., the constraints on the error probabilities are fulfilled with equality. For this case, it has to be shown that the coefficient $C_i^\star$ is non-negative and finite.
We now consider the limit where a coefficient $C_i$ tends to infinity, which is a sequential scheme that never decides in favor of $\Hyp_{i-1}$. By inspecting the limit
\begin{align}\label{eq:gradientLimit}
 \lim_{C_i\rightarrow\infty} \frac{\partial L_\errConstr(C)}{\partial C_i} & =  p(\Hyp_i)\left( 0 - \errConstr_i\right) < 0,
\end{align}
one can see that this contradicts the fact that an infinitely large $C_i$ maximizes \cref{eq:appLagrangianDual}, because a negative gradient would only lead to an optimal value if the corresponding Lagrange multiplier is equal to zero.
Now, the case $C_i=0$ needs closer inspection. In this case the gradient becomes
\begin{align*}
 \frac{\partial L_\errConstr(C)}{\partial C_i}\biggr\rvert_{C_i=0} & =  p(\Hyp_i)\left( \errorDetTilde{0}{i}(\stat{0};\policyOptC) - \errConstr_i\right)\,,
\end{align*}
where $\errorDetTilde{0}{i}(\stat{0};\policyOptC)$ is the actual detection error. The cost of deciding in favor of $\Hyp_{1-i}$ is determined only by $C_{2+i}$. At this point, it has to be recalled that the detection error $\errorDet{0}{i}$ depends on the cost for a wrong detection and the cost for an inaccurate estimate so that the constraint on the detection errors can be implicitly fulfilled by the corresponding estimation constraint.
The two cases in which the actual detection error is smaller and in which the actual detection error is larger than the constraint have to be distinguished. If the actual detection error is smaller than the constraint, i.e., the gradient is negative for a zero Lagrange multiplier, this results in an optimum of $L_\errConstr(C)$ (complementary slackness). Opposed to this, when the actual detection error is larger than the constraint, it results in a positive derivative, which contradicts the assumption that $C_i=0$ maximizes $L_\errConstr(C)$.
Due to the fact that $L_\errConstr(C)$ is concave in $C$, the gradient is positive for $C_i=0$ and that $L_\errConstr(C)$ is decreasing in the limit (see \cref{eq:gradientLimit}), we know that there exist a finite and positive $C_i^\star$ which maximizes $L_\errConstr(C)$. 
With the same line of arguments, it can be shown that this holds also for $i\in\{2,3\}$. That is, the resulting scheme fulfills the requirements on the detection and estimation errors and is, due to the definition of $\rho_n$, of minimum run-length.

It is left to show that the optimal objective is equivalent to the run-length of the scheme. The optimal value of the dual objective can be re-written as:
\begin{align}
   L_\errConstr(C_\errConstr^\star) & = \rho_0(\stat{0};\policyOptCerr) - \sum_{i=0}^1 p(\Hyp_i)C_{i,\errConstr}^\star\errConstr_i - \sum_{i=2}^3 p(\Hyp_{i-2})C_{i,\errConstr}^\star\errConstr_i \nonumber \\
				  & = V^\star(\policyOptCerr) - \sum_{i=0}^1 p(\Hyp_i)C_{i,\errConstr}^\star\errConstr_i - \sum_{i=2}^3 p(\Hyp_{i-2})C_{i,\errConstr}^\star\errConstr_i \nonumber \\
	& = \E[\tau(\stopOptCerr)] + \sum_{i=0}^1 p(\Hyp_i)\bigl( C_{i,\errConstr}^\star \errorDetOpt{0}{i}(\stat{0};\policyOptCerr) + C_{2+i,\errConstr}^\star \errorEstOpt{0}{i}(\stat{0};\policyOptCerr) \bigr ) \nonumber\\
	 & \phantom{ = \E[\tau(\stopOptCerr)] }\; - \sum_{i=0}^1 p(\Hyp_i)\bigl( C_{i,\errConstr}^\star \errConstr_i + C_{2+i,\errConstr}^\star \errConstr_{2+i} \bigr ) \nonumber\\
	 & = \E[\tau(\stopOptCerr)] = \E[\tau(\stopOptErr)] \nonumber
\end{align}
Hence, \cref{eq:optMaxC} is indeed the Lagrangian dual of \cref{prb:constrProblem}.
This concludes the proof.
\section{Regularized Problem Formulation}\label{app:regularizedFormulation}
The regularized formulation of \cref{eq:jointDetEstLPReg} is given by
 \begin{align}
  \begin{split}
   \max_{C\geq0,\rho_n\in\mathcal{L}}\;&\rho_0(\stat{0}) - \sum_{i=0}^1 p(\Hyp_i)C_i\errConstr_i - \sum_{i=2}^3 p(\Hyp_{i-2})C_i\errConstr_i + \frac{\regConst}{N+1}\sum_{n=0}^N\int\rho_n(\stat{n})\dd\mu(\stat{n})\label{eq:jointDetEstLPReg} \\
  \text{s.t.} \quad  & \rho_n(\stat{n}) \leq  \auxVarCostOpt{0,n}(\stat{n}) \quad\text{for } 0\leq n\leq N \\
  & \rho_n(\stat{n}) \leq  \auxVarCostOpt{1,n}(\stat{n}) \quad\text{for } 0\leq n\leq N \\
  & \rho_n(\stat{n}) \leq 1+\int \rho_{n+1}\dd\updateProbMeasure \quad \text{for } 0\leq n < N    
  \end{split}
 \end{align}
where $\regConst$ is a small regularization constant and $\mu(\stat{n})$ is some strictly increasing measure on $(\stateSpaceStat, \metricStat)$. This regularized formulation is used to directly enforce the maximization over the entire domain.
In \cref{eq:jointDetEstLP}, the maximization over the entire domain is only implicit due to the recursive definition of $\rho_n$ and, thus, depends on a sufficiently strong coupling between $\rho_0(\stat{0})$ and $\rho_n(\stat{n})$ through $\updateProbMeasure$. However, this coupling is decreasing with increasing $n$.
First, it has to be shown, that the objective of the regularized problem is still bounded.
For the regularization term without the constant scaling factor, it holds that
\begin{align}
 \sum_{n=0}^N\int\rho_n(\stat{n})\dd\mu(\stat{n}) \leq \sum_{n=0}^N\int g(\stat{n})\dd\mu(\stat{n}) \leq \sum_{n=0}^N\left(\int \auxVarCostOpt{0,n}(\stat{n})\dd\mu(\stat{n}) + \int\auxVarCostOpt{1,n}(\stat{n})\dd\mu(\stat{n})\right)\,.\label{eq:regIneq}
\end{align}
It is assumed, that $\mu(\stat{n})$ has been chosen such that
\begin{align*}
 \int p(\Hyp_i\given\stat{n})\dd\mu(\stat{n}) & = p(\Hyp_i)\,,\quad i\in\{0,1\}\,, \\
 \int p(\Hyp_i\given\stat{n})\Var\left[\paramRV\given\stat{n},\Hyp_i\right]\dd\mu(\stat{n}) & = p(\Hyp_i)\Var\left[\paramRV\given\Hyp_i\right]\,,\quad i\in\{0,1\}\,.
\end{align*}
The integrals over the two auxiliary variables $\auxVarCostOpt{i,n}(\stat{n})$ are hence given by
\begin{align}
\int \auxVarCostOpt{i,n}(\stat{n})\dd\mu(\stat{n}) = C_{1-i}p(\Hyp_{1-i}) + C_{2+i} p(\Hyp_i) \Var\left[\paramRV\given\Hyp_i\right]\,.\label{eq:regAuxInt}
\end{align}
Combining \cref{eq:regIneq} and \cref{eq:regAuxInt} yields an bound on the regularization term
\begin{align}
\frac{\regConst}{N+1}\sum_{n=0}^N\int\rho_n(\stat{n})\dd\mu(\stat{n}) < \regConst \left(\sum_{i=0}^1 C_{i}p(\Hyp_{i}) + C_{2+i} p(\Hyp_i) \Var\left[\paramRV\given\Hyp_{i-2}\right] \right)\,.
\end{align}
Hence, the objective in \cref{eq:jointDetEstLPReg} is upper bounded by
\begin{align}
 \rho_0(\stat{0}) - \sum_{i=0}^1 p(\Hyp_i)C_i(\errConstr_i-\regConst) - \sum_{i=2}^3 p(\Hyp_{i-2})C_i(\errConstr_i-\regConst\Var\left[\paramRV\given\Hyp_i\right])\,.\label{eq:regUpperBoundFinal}
\end{align}
As long as $\errConstr_i-\regConst>0$ for all $i\in\{0,1\}$ and $\errConstr_i-\regConst\Var[\paramRV\given\Hyp_i]>0$ for $i\in\{2,3\}$, boundedness of the objective is guaranteed. That means that the constant has to be chosen according to
\begin{align}
 \regConst < \min\left\{\errConstr_0,\errConstr_1,\frac{\errConstr_2}{\Var[\paramRV\given\Hyp_0]},\frac{\errConstr_3}{\Var[\paramRV\given\Hyp_1]}\right\}\,.
\end{align}
From \cref{eq:regUpperBoundFinal}, it can be seen that the resulting test has lower target error probabilities and estimation errors than the original test and is hence not strictly optimal any more.

\FloatBarrier
\section{Results of the Shift-in-Variance Test with Relaxed Estimation Constraints} \label{app:relaxedShiftInVariance}
In the following, the results of the shift-in-variance test with relaxed estimation constraints under $\Hyp_1$ are presented. The results shown in \cref{tbl:ResultsRandomVarMod,tbl:ResultsRandomVarWeak} differ only in the estimation constraints under $\Hyp_1$. The original problem has an estimation constraint of $0.250$ under $\Hyp_1$, which is first relaxed to $0.300$ and then to $0.500$ in a second step. 
As one can see, all constraints, except for the estimation constraints under $\Hyp_1$, are almost fulfilled with equality by the optimal scheme.
Since the detection constraint under $\Hyp_1$ dominates the corresponding estimation constraint for the example shown in \cref{tbl:ResultsRandomVarWeak}, the empirical estimation errors are below the allowed threshold $\errConstr_3$.

\begin{table}[!ht]
 \centering
 \caption{Shift-in-variance test: Constraints and simulation results - Moderate estimation constraint under $\Hyp_1$.}
  \subtable[Detection and estimation errors]{
 \includegraphics{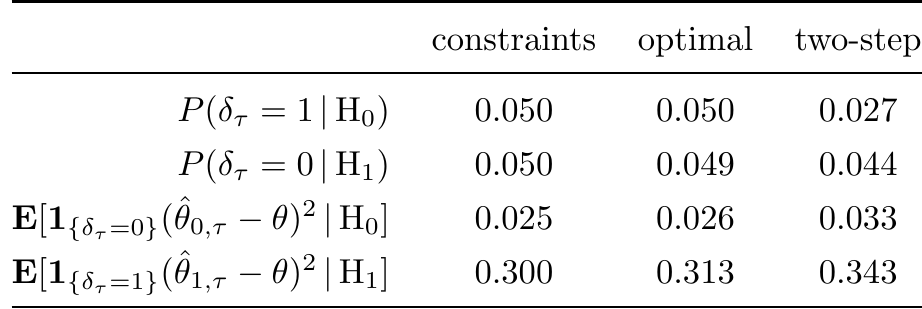} 
 }\\
  \subtable[Run-lengths: The second column contains the expected run-length of the optimal scheme obtained as the output of the \ac{LP}.]{
  \includegraphics{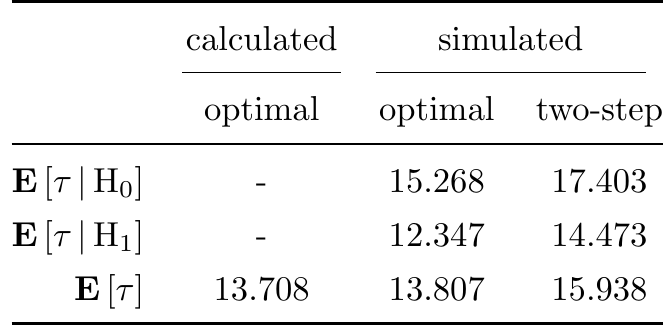}
  }
  \label{tbl:ResultsRandomVarMod}
\end{table}

\begin{table}[!ht]
 \renewcommand{\arraystretch}{1.3}
 \centering
 \caption{Shift-in-variance test: Constraints and simulation results - Weak estimation constraint under $\Hyp_1$.}
 \subtable[Detection and estimation errors]{
\includegraphics{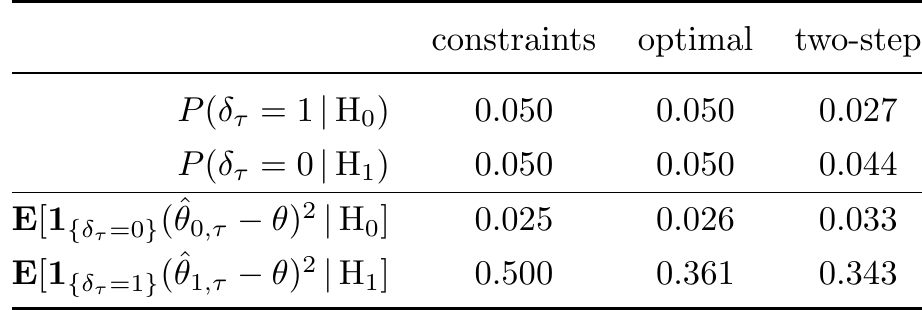}
 }\\
 \subtable[Run-lengths: The second column contains the expected run-length of the optimal scheme obtained as the output of the \ac{LP}]{
\includegraphics{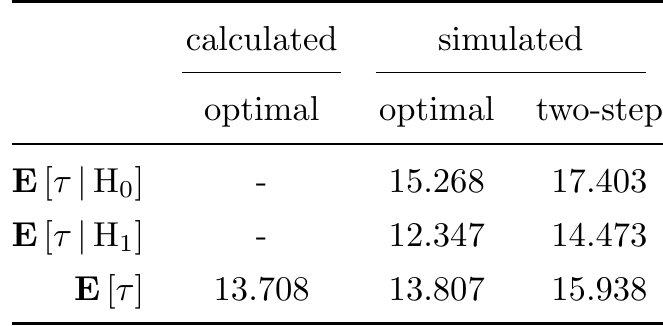}

 }
  \label{tbl:ResultsRandomVarWeak}
\end{table}
    
\FloatBarrier
 
\section*{Acknowledgements}
The authors would like to thank the anonymous reviewers and the editors for their time and effort.

\section*{Funding}

The work of Dominik Reinhard is supported by the German Research Foundation (DFG) under grant number 390542458. 

 \end{document}